\documentclass[review]{elsarticle}

\usepackage{amsmath, amsthm, a4, latexsym, amssymb}

\usepackage{booktabs}
\usepackage[table]{xcolor}
\makeatletter
\def\ps@pprintTitle{%
	\let\@oddhead\@empty
	\let\@evenhead\@empty
	\def\@oddfoot{\reset@font\hfil\thepage\hfil}
	\let\@evenfoot\@oddfoot
}

\usepackage{subfig}
\usepackage{graphicx} 
\usepackage{url}

\usepackage{amsmath,amsthm,amssymb,a4,graphics,color}
\usepackage[mathscr]{eucal} 
\usepackage[active]{srcltx}
\usepackage[usenames,dvipsnames]{pstricks}
\usepackage{epsfig}
\usepackage{pst-grad}
\usepackage{pst-plot} 
\usepackage{textcomp}
\usepackage{dsfont}

\theoremstyle{definition}

\newtheorem{theorem}{Theorem}
\newtheorem{lemma}[theorem]{Lemma}
\newtheorem{rem}[theorem]{Remark}
\newtheorem{cor}[theorem]{Corollary}

\DeclareMathOperator{\EX}{\mathbb{E}}
\DeclareMathOperator{\prob}{\mathbb{P}}

\begin{document}

\begin{frontmatter}

\title{Estimation of the tail index of Pareto-type distributions using regularisation \tnoteref{mytitlenote}}

\vspace{0.2cm}

\author{E. Ocran, R. Minkah,  G.  Kallah-Dagadu and K. Doku-Amponsah \textsuperscript{}  \footnote{Corresponding  Author:  kdoku-amponsah@ug.edu.gh}}
\address{Department of Statistics $\&$ Actuarial Science,\\ School of Physical and Mathematical Sciences, \\University of Ghana}

\begin{abstract}
	In this paper, we introduce reduced-bias estimators for the estimation of the tail index of a Pareto-type distribution. This is achieved through the use of a regularised weighted least squares with an exponential regression model for log-spacings of top order statistics. The asymptotic properties of the proposed estimators are investigated analytically and found to be asymptotically unbiased, consistent and normally distributed. Also, the finite sample behaviour of the estimators are studied through a simulations theory. The proposed estimators were found to yield low bias and MSE. In addition, the proposed estimators are illustrated through the estimation of the tail index of the underlying distribution of claims from the insurance industry.
	
		\textbf{Keywords and Phrases}: 	Statistics of extremes ; tail index; Pareto-type distribution; large deviations; Hill estimator; Strong law  of  large  numbers; Limit  theorem.

		\textbf{MSC 2010}: 62G32,  60F10,  60F25. 
\end{abstract}

\end{frontmatter}
\section{Introduction}
 Pareto-type distributions are often encountered in applications in the area of Finance \cite{longin2016extreme,kithinji2021adjusted,gkillas2018application}, re-insurance \cite{minkah2021robust, minkah2020tail,rohrbeck2018extreme}, risk management \cite{mcneil1997estimating,magnou2017application,afuecheta2020application} and telecommunication \cite{mehrnia2021wireless,finkenstadt2003extreme}. This distribution type has tail function,

\begin{equation}\label{E1}
	1-F(x)=x^{-1/\gamma}\ell(x) \hspace{4pt} \text{as} \hspace{4pt} x\to \infty,
\end{equation}

or equivalently upper tail quantile function 

\begin{equation}\label{E2}
	U(x)=x^{\gamma}\ell_{u}(x) \hspace{4pt} \text{as} \hspace{4pt} x\to \infty.
\end{equation}

The component $\ell_F$ and $\ell_{u}$ are slowly varying functions expressed as 
\begin{equation}\label{E3}
	\lim\limits_{t\to \infty}\dfrac{\ell(xt)}{\ell(t)}=1, \hspace{6pt} t>1.
\end{equation} 

The parameter $\gamma$, is strictly positive for Pareto-type distributions and is also known as the tail index. 

Suppose $X_1, X_2, ..., X_n$ denote independent and identically distributed (i.i.d) random variables drawn from a distribution belonging to the maximum domain of attraction of the Pareto family of distributions, then for some auxiliary sequences of constants $\left\{a_{n}>0; n\ge1\right\}$ and $\left\{b_n; n\ge1\right\}$ \cite{ivette2008tail}

\begin{equation}\label{E4}
	\lim\limits_{n\to \infty}P\left(\dfrac{\max\left\{X_1, X_2, ..., X_n\right\}-b_n}{a_n}\le x\right)=\exp\left\{-\left(1+\gamma x\right)^{-1/\gamma}\right\}, \hspace{4pt} 1+\gamma x\ge0, \hspace{4pt} 	
\end{equation}
where  $\gamma\ge 0.$
The estimation of $\gamma$ continues to receive considerable attention in statistics of extremes as all inferences in extreme value analysis depend on the tail index. In practice, we seek estimators with less variance and bias as possible. A parametric or semi-parametric approach can be employed to estimate the tail index,\cite{Buitendag2018}, \cite{brazauskas2000robust}, \cite{wang2016new}, \cite{tripathi2014improved}. However, in this paper we employ the semi-parametric approach to develop reduced-bias estimators since they result in bias reduction. 

Under the semi-parametric framework, the tail index estimators are dependent on the $k$ largest observations, with these assumption about $k$:
\begin{enumerate}
	\item[$\bullet$] \textbf{Assumption 1:}{ $k(n)\to \infty$ as $n\to \infty$}.
	\vspace{6pt}
	\item[$\bullet$] \textbf{Assumption 2:}{ $k=k(n)= O(n)$ as $n\to \infty$}.
\end{enumerate}
The most widely used semi-parametric tail index estimator is the Hill estimator \cite{Hill1975}.  \cite{Hill1975} approximates the top $k$ order statistics with a Pareto distribution and estimates $\gamma>0$ using a maximum likelihood estimator (MLE). The Hill estimator has the minimum asymptotic variance among the semi-parametric estimators but it is very sensitive to the choice of $k$ \cite{Beirlant2004}. This drawback of the estimator makes its usage challenging in practice, especially in the selection of the tail fraction, $k$. \cite{Hill1975} defined the tail estimator as
\begin{equation}\label{E5}
	\hat{\gamma}_{k}^{H}=\dfrac{1}{k}\sum_{j=1}^{k}\log\left(\dfrac{X_{n-j+1,n}}{X_{n-k,n}}\right).
\end{equation}

 The Hill estimator due to it popularity has received several generalisations: See for example the works of \cite{csorgo1985kernel}, \cite{danielsson1996method}, \cite{beirlant1996excess}, \cite{beran2014harmonic}, \cite{brilhante2013simple}, \cite{paulauskas2017class,paulauskas2013improvement} and \cite{resnick1997smoothing}.

Another estimator which is also a refinement of the Hill estimator is the bias-corrected Hill estimator \cite{Caeiro2005a}. The authors proposed two approaches for reducing the bias of the Hill estimator while maintaining the asymptotic variance of the Hill estimator. Empirically, the bias-corrected Hill estimator yields stable tail index estimates compared to the Hill estimator, i.e., the bias-corrected Hill estimator is less sensitive to the choice of $k$ relative to the Hill estimator.

In this study, we seek to propose alternative tail index estimators which empirically yield much more stable tail index estimates and attain the minimum asymptotic variance of the Hill estimator under some conditions. Tail index estimators which yield stable estimates are highly sort after in practice since they alleviate the problem of selecting an optimal $k$ to some extent.

\section{\textbf{Estimation Methods}}\label{S2}

We let $X_1, X_2, X_3, ... ,X_n$ denote a sequence of i.i.d random variables drawn from a population with distribution function $F$ and the associated tail quantile function $U$. Let $X_{1,n}\le X_{2,n}\le ... \le X_{n,n}$ be the order statistics associated with the sample. Using Eqn. (\ref{E2}), the order statistics can be jointly expressed as, 
\begin{equation}\label{E6}
\log X_{n-j+1,n}\stackrel{d}\sim\gamma\log U^{-1}_{j,n}+\log \ell\left(U^{-1}_{j,n}\right),
\end{equation}
where $U^{-1}_{j,n}$, $j=1,2,3,...,n$ represent the order statistics of the standard uniform distribution. Using Eqn. (\ref{E6}), we \cite{Beirlant1999} demonstrated that

\begin{equation}\label{E7}
\dfrac{\log X_{n-j+1,n}}{\log X_{n-k,n}}\stackrel{d}\sim \gamma \log \dfrac{U^{-1}_{k+1,n}}{U^{-1}_{j,n}}+\log \ell\left(\dfrac{U^{-1}_{j,n}}{U^{-1}_{k+1,n}}\right)
\end{equation}
$k\in\left\{2,3, ..., n-1 \right\}$ and we also obtained a more refined expression of Eqn. (\ref{E7}) by imposing a second-order assumption on the rate of convergence to Eqn. (\ref{E3}). This is stated in as an  assumption as follows:

\textbf{Assumption 3}: There exists a real constant $\rho\le 0$ and a rate function $b$ satisfying $b(x)\to 0$ as $x\to \infty$ such that for all $u\ge 1$,
\begin{equation}\label{E8}
\lim_{x\to\infty}\Big[	\log \dfrac{\ell(xu)}{\ell(u)}\Big]= b(t)h_{\rho}(u),
\end{equation}
with $h_{\rho}(u)=\int_{1}^{u}y^{\rho-1}dy$ \cite{Beirlant1999}.
 
 Under assumption 3, \cite{Beirlant1999} showed that the weighted log-spacings of order statistics,
 
\begin{equation}\label{E9}
	T_j=j\left\{\log X_{n-j+1,n}-\log X_{n-j,n}\right\}, \hspace{4pt} 1\le j \le k<n,
\end{equation}
are approximately exponentially distributed. They particularly obtained the expression

\begin{equation}\label{E10}
	T_j=\left(\gamma+b_{n,k}\left(\dfrac{j}{k+1}\right)^{-\rho}\right)f_{j}, \hspace{4pt} 1\le j \le k,
\end{equation} 
where $b_{n,k}=b\left((n+1)/(k+1)\right) \to 0$ as $k, n \to \infty$, and $f_{j}'s$ are i.i.d exponentially distributed with a unit mean and $\rho$ ($\rho<0$)  is  a second-order parameter. The authors employed MLE to the estimate the parameters in Eqn. (\ref{E10}).

Using Eqn. (\ref{E10}) and assumptions 1 and 2, \cite{Beirlant2002}, demonstrated that $T_{j}$ can be further approximated as a regression model,
\begin{equation}\label{E11}
	T_{j}=T_j(\rho,\epsilon_j)=\gamma+b_{n,k}C_{j}(\rho)+\epsilon_{j} \hspace{4pt}\text{for}\hspace{4pt} j\in\left\{1, 2, ..., k\right\},
\end{equation}
where $b_{n,k}=b\left(n/k\right)$ is the slope, $C_{j}=C_j(\rho)=\left(j/(k+1)\right)^{-\rho}$ is the covariate, $\gamma$ is the intercept and $\epsilon_{j}$ are error terms with asymptotic mean, $0$, and variance, $\gamma^2$.

\cite{Beirlant2002} proposed the ordinary least squares estimator for the estimation of $\gamma$ in Eqn. (\ref{E11}). Further,  based on Eqn. (\ref{E11}), \cite{Buitendag2018} have introduced the ridge regression estimator for estimating $\gamma$. In this paper, we propose the regularised weighted least squares estimators for estimating $\gamma$ in Eqn. (\ref{E11}).

\subsection{The proposed estimators}
In order to estimate $\gamma$, the loss function of the regularised weighted least squares for Eqn. (\ref{E11}) is defined as
\begin{equation}\label{E12}
	L_{k}\left(\gamma,b_{n,k};\lambda,W \right) =\sum_{j=1}^{k}W_{j}\left(T_{j} -\gamma-b_{n,k}C_{j}\right)^2+\lambda k b_{n,k}^2 , \hspace{10pt}\lambda\ge 0.
\end{equation}

Here, $W_j$ is the weight function defined as

\begin{equation}\label{E13}
	W_{j}=\left(1-\theta_{j}^{\alpha(k)}\dfrac{j}{k+1}\right), j\in\{1, 2, ..., k\},
\end{equation}

where $\theta_{j}\stackrel{i.i.d }\sim U(0, 1)$. Thus, $W_{j}\in (0, 1)$ and decreases linearly with respect to $j$. The exponent $\alpha(k)\ge 0$ is chosen such  that 
\begin{equation}\label{E14}
	\varDelta:=\lim\limits_{k\to \infty}\dfrac{\alpha(k)}{k}<\infty.
\end{equation}
In this study, we consider $0<\varDelta\le1$. 

 Thus, we would define $\alpha(k)$ such that $0<\alpha(k)\le k$. Note that, $W_{j}$ is random through $\theta_{j}$, and when the exponent is $0$, $W_{j}$ is deterministic. In particular, when $\alpha(k)=0$, we obtain the weight function $g_{j}=1-j/(k+1)$ as introduced by \cite{ocran2021reduced}. Nevertheless, we can approximate the weight $g_{j}$ as a limit of the current result by allowing $\varDelta$ to approach $0$. 

We minimize  the loss function $L_{k}$ with respect to $\gamma$ and $b_{n,k}$ to obtain jointly estimates $\hat\gamma$ and $\hat b_{n,k}$ 

\begin{equation}\label{E15}
	\hat{b}_{n,k}(\lambda, \theta)=\dfrac{\sum_{j=1}^{k}\tilde{W}_{j}\left(C_{j}-\sum_{j=1}^{k}\tilde{W}_{j}C_{j}\right)T_{j}}{2\kappa(\alpha(k))\lambda+\sum_{j=1}^{k}\tilde{W}_{j}C_{j}^{2}-\left(\sum_{j=1}^{k}\tilde{W}_{j}C_{j}\right)^2} 
\end{equation}
and
\begin{equation}\label{E16}
	\hat{\gamma}_{RW}(\lambda, \theta)=\sum_{j=1}^{k}\tilde{W}_{j}T_{j}-\hat{b}_{n,k}(\lambda,\theta)\sum_{j=1}^{k}\tilde{W}_{j}C_{j}.
\end{equation}
while,
\begin{equation}\label{E17}
	\tilde{W}_{j}=\dfrac{W_{j}}{\sum_{j=1}^{k}W_{j}}, 1\le j\le k,
\end{equation}
\begin{equation}\label{E18}
	C_{j}=\left(\dfrac{j}{k+1}\right)^{-\rho}, \hspace{12pt} \text{for} \hspace{12pt} j \in (1,2,...,k), \hspace{12pt} \rho<0,
\end{equation}
and
\begin{equation}\label{E19}
	\kappa(\alpha(k))=\dfrac{\alpha(k)+1}{2\alpha(k)+1}.
\end{equation}

We substitute  Eqn. (\ref{E13}) into Eqn. (\ref{E17}) to  obtain   an explicit expression for Eqn. (\ref{E17}) as
\begin{equation}\label{E20}
	\tilde{W}_{j}\approxeq\dfrac{2\kappa(\alpha(k))}{k}\left(1-\theta_{j}^{\alpha(k)}\dfrac{j}{k+1}\right).
\end{equation}
The parameter $\rho<0$ is estimated using the minimum variance approach introduced in \cite{Buitendag2018}.

In addition,  the parameter $\lambda$ in Eqn. (\ref{E12}) is the penalty that regulates the bias coefficient $b_{n,k}$. The loss function, $L_{k}$, minimises the weighted sum of squared residuals and also regulates the size of the bias coefficient $b_{n,k}$. The penalty term shrinks the bias term, $b_{n,k}$ to $0$ as the penalty parameter, $\lambda$,  increases. Thus, the larger the value of $\lambda$, the higher the contribution of the penalty term to the loss function and the stronger the regularisation process. To obtain an estimator for the penalty term, $\lambda$, we minimise the mean squared error (MSE) of the proposed estimator, $\hat{\gamma}_{RW}(\lambda)$.  See,example \cite{Buitendag2018}.

Note that, since the weight function depends on the  $\theta_{j}$'s, the estimators in Eqn. (\ref{E15}) and Eqn. (\ref{E16}) also depend on the $\theta_{j}$'s through the weight function. Therefore, we would find the MSE by conditioning on the $\theta_{j}$'s. From Eqn. (\ref{E15}) and Eqn. (\ref{E16}), the MSE for $\hat{\gamma}_{RW}(\lambda)$ is obtained as

\begin{equation}\label{E22}
	S_{1}(\theta)=\sum_{j=1}^{k}\tilde{W}_{j}C_{j},
\end{equation}

\begin{equation}\label{E23}
	S_{2}(\theta)=\sum_{j=1}^{k}\tilde{W}_{j}C_{j}^2-\left(\sum_{j=1}^{k}\tilde{W}_{j}C_{j}\right)^2,
\end{equation}

\begin{equation}\label{E24}
	\dot{S}(\theta)=\sum_{j=1}^{k}\tilde{W}_{j}^2\left(S_{1}(\theta)-C_{j}\right),
\end{equation}

\begin{equation}\label{E25}
	\ddot{S}(\theta)=\sum_{j=1}^{k}\tilde{W}_{j}^2\left(S_{1}(\theta)-C_{j}\right)^2,
\end{equation}
and
\begin{equation}\label{E26}
	\phi(\alpha(k))=\dfrac{\left(\alpha(k)+1\right)\left(6\alpha^2(k)+4\alpha(k)+1\right)}{\left(2\alpha(k)+1\right)^2}.
\end{equation}

We now derive an expression for the penalty term, $\lambda$. Minimising $MSE\left(\hat{\gamma}_{RW}(\lambda)|\theta_{j}=\theta\right)$ over $\lambda$, the optimal value of $\lambda$ is obtained by solving the equation
\begin{equation}\label{E27}
	2b_{n,k}^2\kappa(\alpha(k))S_{1}(\theta)S_{2}(\theta)\lambda-2\kappa(\alpha(k))\gamma^2\dot{S}(\theta)\lambda-\gamma^{2}S_{2}(\theta)\dot{S}(\theta)-\gamma^2S_{1}(\theta)\ddot{S}(\theta)=0.
\end{equation}
we  obtain, 
\begin{equation}\label{E28}
	\lambda_{k}(\theta)=\dfrac{\gamma^2S_{2}(\theta)\dot{S}(\theta)+\gamma^2S_{1}(\theta)\ddot{S}(\theta)}{2b_{n,k}^2\kappa(\alpha(k))S_{1}(\theta)S_{2}(\theta)-2\kappa(\alpha(k))\gamma^2\dot{S}(\theta)}.
\end{equation}

In order to estimate $\hat{\lambda}_{k}(\theta)$, we assume the slowly varying function $\ell$ in (\ref{E8}) is constant.Thus,  we  have  
\begin{equation}\label{E29}
	b(x)=\gamma\beta x^{\rho}\{1+O(1)\} \hspace{6pt} \text{as} \hspace{6pt} k\to \infty
\end{equation}
for some $\beta \in \mathbb{R}$. We estimate $\beta$ via the estimator proposed by \cite{gomesa2002asymptotically}  to  obtain 

\begin{equation}\label{E31}
	\hat{\lambda}_{k}(\theta)=\dfrac{S_{1}(\theta)\ddot{S}(\theta)+\dot{S}(\theta)S_{2}(\theta)}{2\kappa(\alpha(k))S_{1}(\theta)S_{2}(\theta)\hat{\beta}_{k}^2\left(\dfrac{n}{k}\right)^{2\hat{\rho}}-2\kappa(\alpha(k))\dot{S}(\theta)}.
\end{equation}

The penalty term, $\lambda$ is required to be non-negative; therefore we define $\hat{\lambda}_{k}^{+}(\theta)=\max\{\hat{\lambda}_{k}(\theta),0\}$. We  the  obtain a penalty term  and  the  estimators  which does not depend on $\theta_{j}$'s,by averaging $\hat{\lambda}_{k}(\theta)$ over the $\theta_{j}$'s,  as  follows:

\begin{equation}\label{E32}
	\hat{\lambda}_{k}=\EX_{\theta}\left\{\hat{\lambda}_{k}(\theta)\mathds{1}_{\left\{S_{1}(\theta)S_{2}(\theta)\hat{\beta}_{k}(n/k)^{2\hat{\rho}}>\dot{S}(\theta)\right\}} \right\},   
\end{equation}

 and  we   defined  the  proposed  estiamtor  of  $\gamma$  by $$\gamma_{RW}(\lambda):=\EX_{\theta}(\gamma_{RW}(\lambda,\theta)).$$

\subsubsection{Asymptotic Properties of the Proposed estimators}
Unbiasedness, consistency and normality are desirable properties of a good estimator.  In this section, we investigate desirable properties of  the proposed estimator possesses.

We shall summarise the asymptotic  behaviour  of the statistics used to build the MSE of the proposed estimator  in  Lemma~\ref{L1}. These properties will be required in the proof of the asymptotic consistency and sampling distribution of the proposed estimator. Henceforth, anytime  we  use  the  term  $a.s$  it  is  with  respect  to  the  law  of  the  i.i.d sequence $\theta_1,\theta_2, \theta_3,...,\theta_k$.  
\begin{lemma}\label{L1}
	Assume that $\rho$ is estimated by a consistent estimator $\hat{\rho}$ and (\ref{E14}) holds, then as $k\to \infty$ and $k/n\to 0$;
	\begin{enumerate}
		\item[i. ] $S_{1}(\theta)=\sum_{j=1}^{k}\tilde{W}_{j}C_{j}\stackrel{a.s}\longrightarrow1/(1-\rho).$
		\item[ii.] $S_{2}(\theta)=\sum_{j=1}^{k}\tilde{W}_{j}C_{j}^2-\left(\sum_{j=1}^{k}\tilde{W}_{j}C_{j}\right)^2\stackrel{a.s}\longrightarrow\rho^2/(1-2\rho)(1-\rho)^2.$
		\item[iii.] $\dot{S}(\theta)=\sum_{j=1}^{k}\tilde{W}_{j}^2\left(S_{1}(\theta)-C_{j}\right)\stackrel{a.s}\longrightarrow0$
		\item[iv.] $\ddot{S}(\theta)=\sum_{j=1}^{k}\tilde{W}_{j}^2\left(S_{1}(\theta)-C_{j}\right)^2\stackrel{a.s}\longrightarrow0$.
	\end{enumerate}
	
\end{lemma}

\begin{lemma}\label{L2}
	Suppose $\rho<0$ and $\beta\in \mathbb{R}$ are estimated by their respective consistent estimators $\hat{\rho}$ and $\hat{\beta}_{k}$, then as $k,n \to \infty$ and $k/n \to 0$,
	\begin{equation}
		\hat{\lambda}_{k}(\theta)\stackrel{a.s}\longrightarrow 0.
	\end{equation}
\end{lemma}

It follows from Lemma~\ref{L2} that the regularised weighted least estimator, $\hat{\gamma}_{RW}(\lambda)$ is asymptotically unbiased. That is, as $k\to \infty$, \hspace{4pt} $bias\left(\hat{\gamma}_{RW}(\lambda)\right) \to 0$.
The bias of the proposed estimator is given as,

\begin{equation*}
bias\left(\hat{\gamma}_{RW}(\lambda)\right)=\EX_{\theta}\left\{\dfrac{2b_{n,k}\kappa(\alpha(k))S_{1}(\theta)\lambda(\theta)}{2\kappa(\alpha(k))\lambda(\theta)+S_{2}(\theta)}\right\}.
\end{equation*}
where 
\begin{equation*}
	bias\left(\hat{\gamma}_{RW}(\lambda)|\theta=\hat\theta\right)=\dfrac{2b_{n,k}\kappa(\alpha(k))S_{1}(\hat\theta)\lambda(\theta)}{2\kappa(\alpha(k))\lambda(\hat\theta)+S_{2}(\hat\theta)}.
\end{equation*}

Since the term in the bracket converges to $0$ as $k\to \infty$ almost sure by Lemma~\ref{L2}, the expectation of the term will converge to $0$ as $k\to \infty$.
Therefore,
\begin{align*}
	\lim\limits_{k\to \infty}bias\left(\hat{\gamma}_{RW}(\lambda)\right)
	&=\lim\limits_{k \to \infty}\int_{[0,1]^{k}}\left\{\dfrac{2b_{n,k}\kappa(\alpha(k))S_{1}(\theta)\lambda(\theta)}{2\kappa(\alpha(k))\lambda(\theta)+S_{2}(\theta)}\right\}d\theta\\
	&=\int_{[0,1]^{\infty}}\left\{\lim\limits_{k \to \infty}\dfrac{2b_{n,k}\kappa(\alpha(k))S_{1}(\theta)\lambda(\theta)}{2\kappa(\alpha(k))\lambda(\theta)+S_{2}(\theta)}\right\}d\theta\\\
	&=0 .
\end{align*}

which  gives , 
 $\lim_{k\to\infty}bias(\hat{\gamma}_{RW}(\lambda))= 0$.

Similarly, we  may  use   Lemma~\ref{L1} and Lemma~\ref{L2},to show that the MSE of $\hat{\gamma}_{RW}(\lambda) \to 0$ as $k \to \infty$. 
Now, observe  that 
\begin{equation*}
	MSE\left(\hat{\gamma}_{RW}(\lambda)\right)=\EX_{\theta}\left(MSE\left(\hat{\gamma}_{RW}(\lambda)|\hat\theta=\theta\right)\right).
\end{equation*}
and  therefore,  we  have  
\begin{align*}
	\lim\limits_{k \to \infty}MSE\left(\hat{\gamma}_{RW}(\lambda)\right)
	&=\lim\limits_{k \to \infty}\int_{[0,1]^k}\left\{(MSE\left(\hat{\gamma}_{RW}(\lambda)|\hat \theta=\theta\right)\right\}d\theta\\
	&=\int_{[0,1]^\infty}\left\{\lim\limits_{k\to \infty}(MSE\left(\hat{\gamma}_{RW}(\lambda)|\hat \theta=\theta\right)\right\}d\theta\\
	&=0.
\end{align*}
This implies that the proposed estimator is asymptotically consistent under some conditions.

\begin{theorem}\label{TH1}
	Suppose Eqn.(\ref{E3}), Eqn.(\ref{E8}) and Eqn.(\ref{E14}) are satisfied. Assume also that $\rho$ is estimated by a consistent estimator $\hat{\rho}$  with  $\EX\Big[(1-2\hat\rho)(1-\hat\rho)^2/\hat\rho^2\Big]<\infty$. Then if assumptions $A_1$ and $A_2$ holds, and $\sqrt{k}b_{n,k}\rightarrow_{p}0$,  then  we  have  
	\begin{equation*}
		\sqrt{k}\left(\EX_{\theta}(\hat{\gamma}_{RW}(\lambda,\theta))-\gamma\right)\rightarrow_{d}
		N(0, \gamma^2).
	\end{equation*}
\end{theorem}

\begin{cor}\label{COR2}
	Suppose (\ref{E3}) and (\ref{E8}) holds and also $\lim\limits_{k\to \infty}\alpha(k)/k\to\varDelta\in(0,\infty)$. Assume $\rho$ is estimated by a consistent estimator $\hat{\rho}$  with $\EX\Big[(1-2\hat\rho)(1-\hat\rho)^2/\hat\rho^2\Big]<\infty.$ Then as $k, n \to \infty$ and $\sqrt{k}b_{n,k}\longrightarrow_{p} 0$,
	
	$$\sqrt{k}\left(\EX_{\theta}(\hat{\gamma}_{RW}(\lambda,\theta))-\gamma\right)/\gamma\to_d N(0,1)$$	
\end{cor}

\begin{proof}
	The proof of Corollary~\ref{COR2} follow from the proof of Theorem~\ref{TH1}.
\end{proof}

Theorem~\ref{TH1} discusses the asymptotic normality of $\hat{\gamma}_{RW}(\lambda)$ defined in Eqn.(\ref{E16}). To prove Theorem~\ref{TH1} we require the following properties in addition.

We  write  
\begin{equation}\label{E34}
	\mathfrak{M}_{\lambda}(\theta):=\dfrac{\tilde{W}_{j}(C_{j}-\sum_{j=1}^{k}\tilde{W}_{j}C_{j})}{2\kappa(\alpha(k))\lambda+\sum_{j=1}^{k}\tilde{W}_{j}C_{j}^2-\left(\sum_{j=1}^{k}\tilde{W}_{j}C_{j}\right)^2}.
\end{equation}

\begin{lemma}\label{L3}
	Let  $C_{j}=\left(\dfrac{j}{k+1}\right)^{-\rho}, \tilde{W}_{j}=\dfrac{W_{j}}{\sum_{j=1}^{k}W_{j}}, j\in\{1,2,...,k\}$ and $\rho<0.$   Then   as $k\to \infty$,   
	$$\mathfrak{M}_{\lambda}(\theta)\to O(1/k^{1+\omega}),$$
	where   $0<\omega\le 0.1 $.
\end{lemma}

Lemma~\ref{L3} is required in the proof of Lemma~\ref{L4}.

\begin{lemma}\label{L4}
	Let $T_{1}, T_{2}, T_{3},...$ be independent random variables from an exponential distribution with mean $\mu_{i}<\infty$, for all $i$. Then for any $\epsilon>0$ 
	\begin{equation}\label{E35}
		\lim\limits_{k \rightarrow \infty}\prob\left(\sqrt{k}\hat b_{n,k}\ge \epsilon\right)=0,
	\end{equation}
	where $\hat b_{n,k}$ is defined by Eqn.(\ref{E15}).
\end{lemma}

\begin{rem}
Lemma~\ref{L4} shows the statistics $\sqrt{k}\hat b_{n,k} \to_p 0$ as $k\to \infty$.
\end{rem} 

The next Lemma is about the satisfaction of the Lyapunov's version of the central limit theorem. The Lyapunov's variant of the central limit theorem assumes the existence of a finite moment of an order higher than two.
\begin{lemma}\label{L5}
	Suppose that $Z_{1}, Z_{2},...$ are independent random variables such that $\EX(Z_{k})=\hat{\mu}_{k}$ and $Var(Z_{k})=\sigma_{k}^2<\infty$, then there exists $\delta>0$ such that
	\begin{equation}\label{C3E44}
		\lim\limits_{k \rightarrow \infty}\dfrac{1}{s_{n}^{2+\delta}}\sum_{j=1}^{k}\EX\left(|Y_{j}-\hat{\mu}_{j}|^{2+\delta}\right)=0,
	\end{equation}
	where $Z_{k}=\EX_{\theta}[\tilde{W}_{k}(\theta_{k})]T_{k}.$
\end{lemma}

\begin{rem}
	Setting the penalty term to 0 reduces the  regularised weighted least squares estimator to a weighted least squares estimator. The difference between this weighted least squares estimator and the one introduced by \cite{ocran2021reduced} is that, this weighted least squares estimator has smaller asymptotic variance and this is due to the introduction of randomness into the weight function. The resulting weighted least squares estimator is also asymptotically unbiased, consistent, and normally distributed with mean $0$ and variance $\gamma^2$. 
\end{rem}

 For  proofs of   the above  Lemmas  we  refer  interested person  to    \cite{ocran2021reduced}.

\section{Simulation Study}
In the previous section, we proposed the regularised weighted least squares estimators under the semi-parametric setting to estimate the tail index of the underlying distribution of a given data from the Pareto-type of distributions. In this section, we perform a simulation study to compare the performance of our proposed estimators to other existing semi-parametric tail index estimators. Particularly, the regularised weighted least squares, RWLS, the reduced-bias weighted least squares with modified weight function, WLS, the ridge regression, RR \cite{Buitendag2018}, the least squares, LS \cite{Beirlant2002}, the Hill estimator, HILL  \cite{Hill1975} and the bias-corrected Hill, BCHILL \cite{Caeiro2005a} in the case of Pareto-type distributions are compared.

\subsection{Simulation design}
We consider the Fr\'{e}chet and Burr XII from the Pareto-type distributions as shown in Table~\ref{Table:T1}. For each distribution $F$, we generate 1000 repetitions of samples of size $n=50, 500$ and $2000$. For the Fr\'{e}chet distributions, we consider $\alpha=10, 2$ and $1.0$; and for the Burr XII we consider the mixtures 
\begin{enumerate}
	\item[i.] $\xi=\sqrt{10} \hspace{6pt} \&  \hspace{6pt}\tau=\sqrt{10}$,
	\item[ii.] $\xi=\sqrt{2} \hspace{6pt} \& \hspace{6pt} \tau=\sqrt{2}$ and
	\item[iii.] $\xi=2 \hspace{6pt} \& \hspace{6pt} \tau=1/2$
\end{enumerate}

 to obtain the tail index values $\gamma=0.1, 0.5$ and $1.0$, respectively. We consider the finite sample behaviour of the proposed estimators, RWLS and WLS, and also compared these estimators with RR, LS, BCHILL and HILL. The MSE and bias are plotted as a function of the number of top-order statistics, $k$, to investigate the estimators' sample path behaviour.

In the case of the weight function, the $\theta_{j}$'s will be replaced with their point estimate, in this case, the mean of a standard uniform distribution. In the case of $\alpha(k)$, we select $\alpha(k)$ such that $\varDelta \to 0$, as $k\to \infty$. This choice of $\alpha(k)$ is made because in practice we have observed that it yields much more stable estimates compared to when $\alpha(k)$ is selected such that $\varDelta \to c>0$ in application.

\begin{table}[!htbp]
	\caption{Heavy-tailed distributions from the Pareto-type distribution}
	\centering
	\begin{tabular}{*{7}{c}}
		\hline
		Distribution & $1-F(x)$ & $\ell_{F}(x)$ & $\gamma$\\ \hline 
		Burr type XII & $\left(1+x^{\tau}\right)^{-\xi}$ & $\left(1+x^{-\tau}\right)^{-\xi}$ & $1/\tau\xi$ \\
		Fr\'{e}chet & $1-F(x)=1-exp(-x^{-\alpha})$ & $1-\frac{x^{-\alpha}}{2}+O(x^{-\alpha})$ & $1/\alpha$ \\ \hline
	\end{tabular}
	\label{Table:T1}
\end{table}

\subsection{Discussion of simulation results}
In this section, we discuss the behaviour of RWLS and WLS relative to RR, LS, HILL and BCHILL. The Mean Square Error (MSE) and bias are the performance measures in the simulation studies. The simulation results for the Burr distribution with different tail indexes are shown in Figures \ref{Fig:Figure 1} - \ref{Fig:Figure 3}. Also, Figures \ref{Fig:Figure 4} - \ref{Fig:Figure 6}  present the simulation results for the Fr\'{e}chet distribution with varying tail indexes.

From these figures, the plots of WLS and RWLS follow the same sample path for $k\le0.4n$, i.e., their performance are relatively the same on that interval. WLS and LS are very close to each other, though generally, WLS slightly outperforms LS in terms of MSE and bias. Thus, generally the WLS can be considered the most appropriate estimator of the tail index among the regression-based estimators (i.e., RR, LS, WLS and RWLS) since it mostly has smallest bias and MSE across all samples.

Additionally, the MSE plots of the proposed estimators are low and near constant over the central part of $k$, except in the case of Burr XII with $\gamma=1.0$. With the exception of the HILL estimator (which globally has the highest MSE),  the MSE curves of the estimators are mostly close to each other in the central $k$ region, especially in the case of the Fr\'{e}chet distribution. This implies that the proposed estimators are competitive with the existing estimators. However, the proposed estimators, WLS and RWLS, generally attain the lowest bias for small samples, i.e., $n=50$. Furthermore, for medium to large values of $k$, the sample paths of RWLS in the MSE and bias plots are between HILL and RR. Even though the BCHILL estimator mostly has the smallest MSE and bias, the proposed estimators (RWLS and WLS) outperform it for large values of $k$.

Hence, from the simulation results, WLS and RWLS are appropriate for the estimation of the tail index of the Pareto-type distributions in terms of MSE and bias.

\begin{figure}[!htbp]
	\centering
	\subfloat{\includegraphics[width=4cm,height=2.5cm]{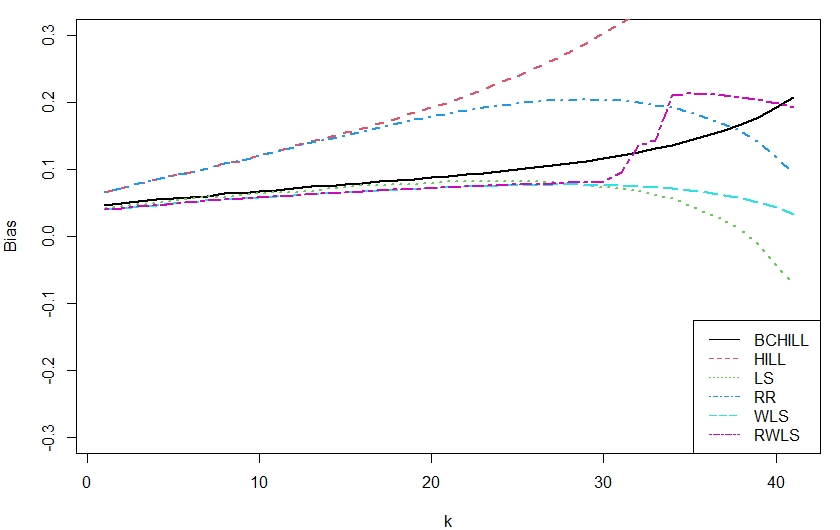}}\hfill
	\subfloat{\includegraphics[width=4cm,height=2.5cm]{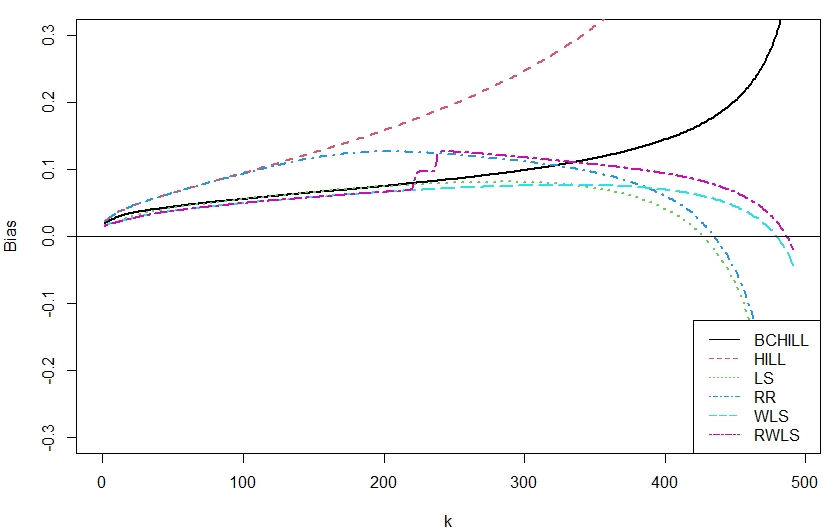}}\hfill
	\subfloat{\includegraphics[width=4cm,height=2.5cm]{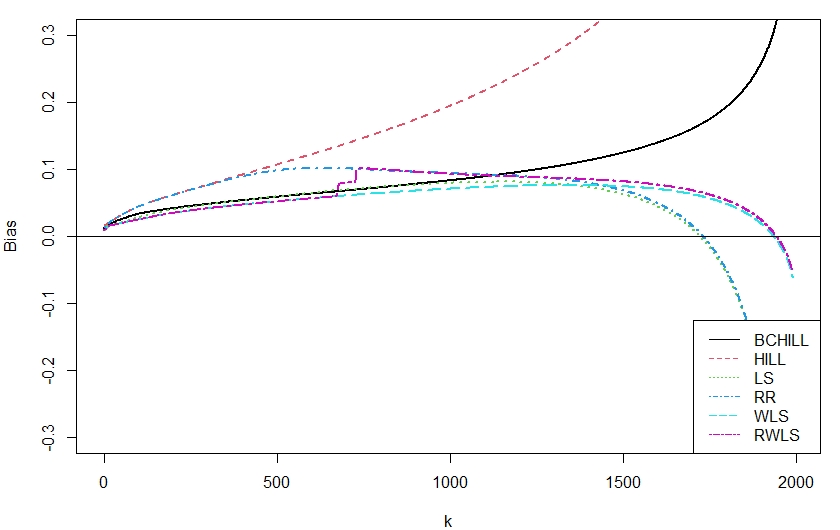}}\hfill
	\\[\smallskipamount]
	\subfloat{\includegraphics[width=4cm,height=2.5cm]{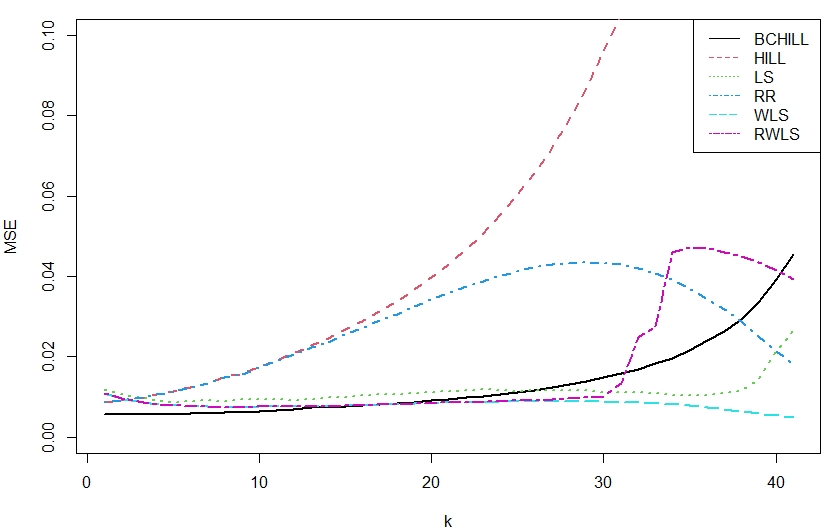}}\hfill
	\subfloat{\includegraphics[width=4cm,height=2.5cm]{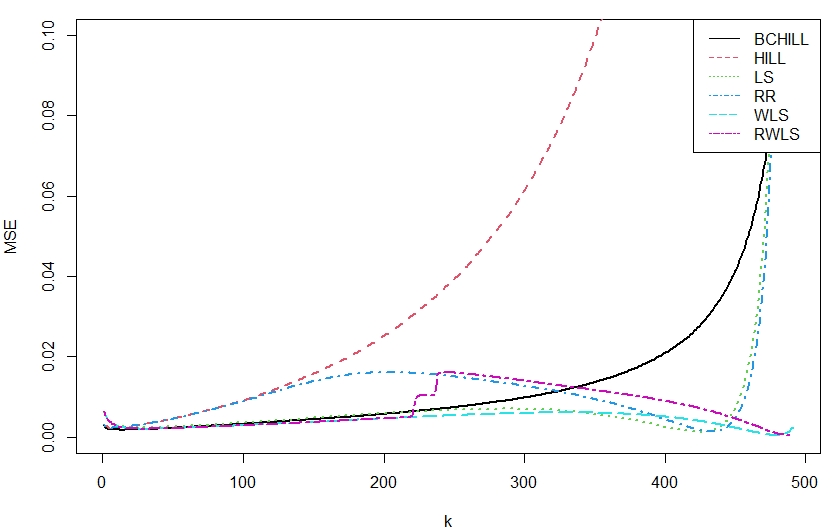}}\hfill
	\subfloat{\includegraphics[width=4cm,height=2.5cm]{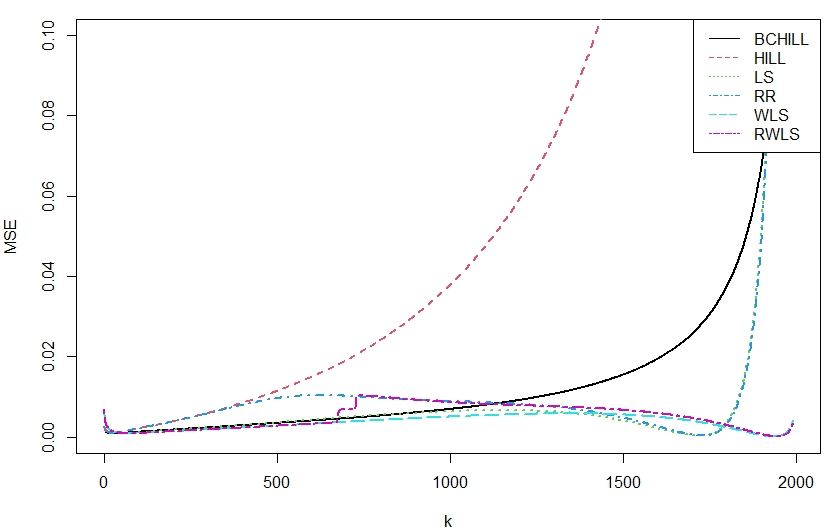}}
	\caption{Results for Burr type XII distribution with $\gamma=0.1$: Bias(top row) and MSE(bottom row). First column: $n=50$; second column: $n=500$; and third column: $n=2000$.}
	\label{Fig:Figure 1}
\end{figure}

\begin{figure}[!htbp]
	\centering
   \subfloat{\includegraphics[width=4cm,height=2.5cm]{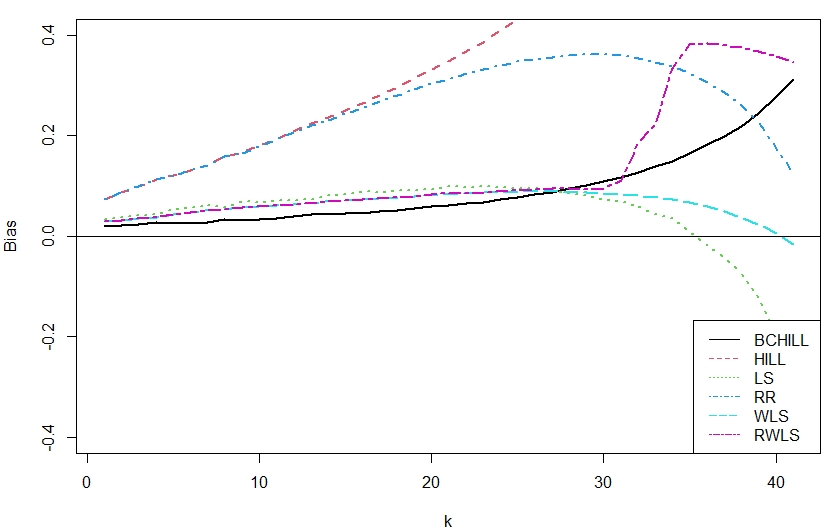}}\hfill
   \subfloat{\includegraphics[width=4cm,height=2.5cm]{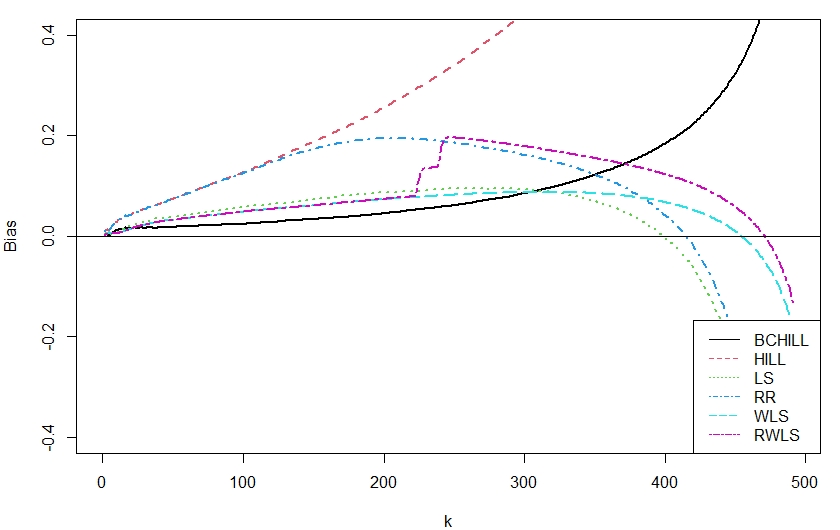}}\hfill
   \subfloat{\includegraphics[width=4cm,height=2.5cm]{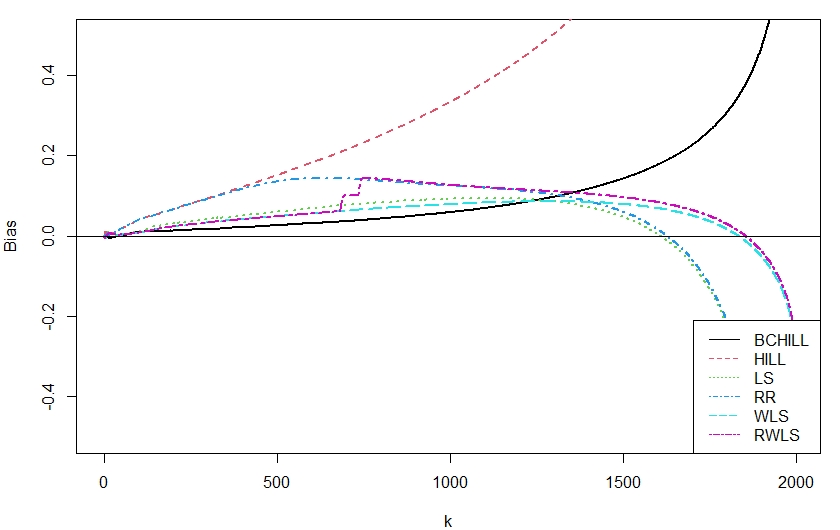}}\hfill
	\\[\smallskipamount]
	\subfloat{\includegraphics[width=4cm,height=2.5cm]{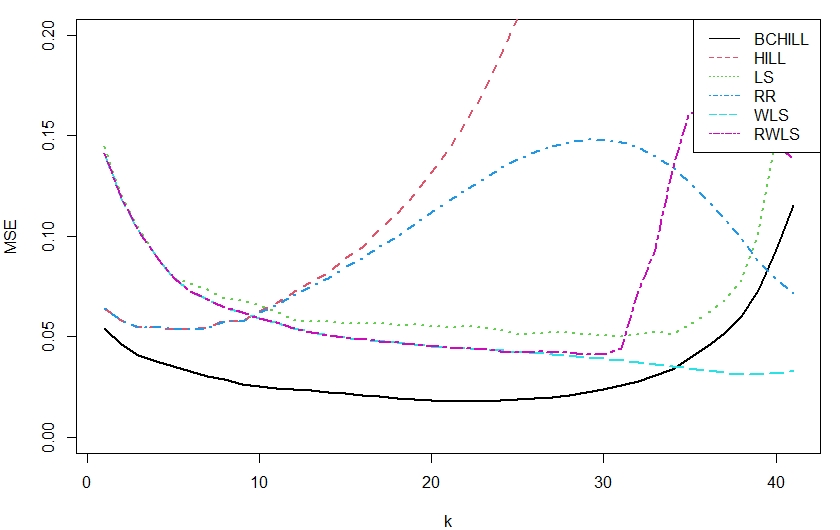}}\hfill
	\subfloat{\includegraphics[width=4cm,height=2.5cm]{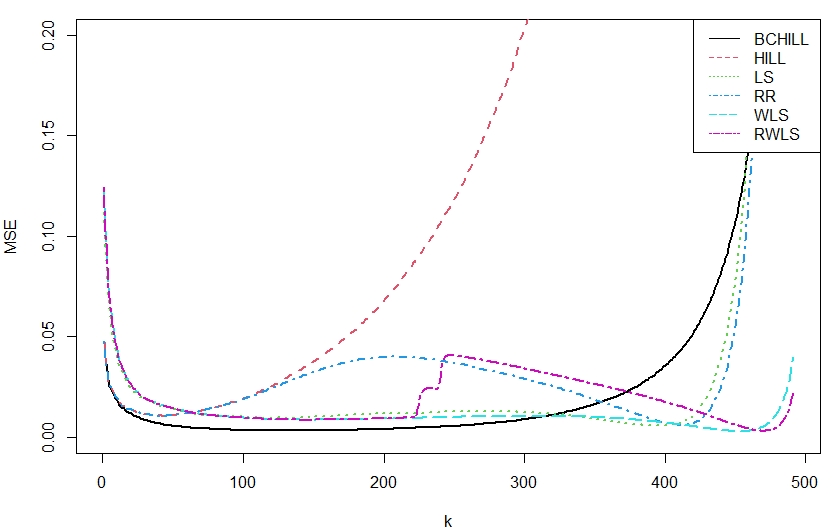}}\hfill
	\subfloat{\includegraphics[width=4cm,height=2.5cm]{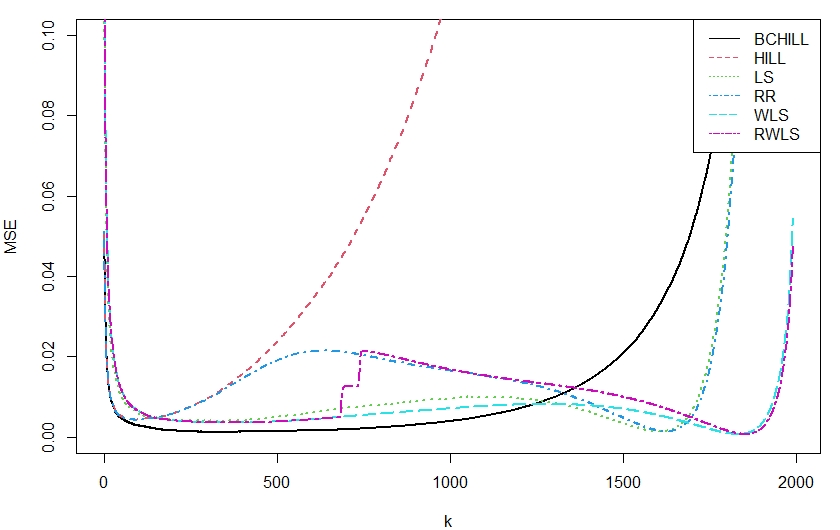}}
	\caption{Results for Burr type XII distribution with $\gamma=0.5$: Bias(top row) and MSE(bottom row). First column: $n=50$; second column: $n=500$; and third column: $n=2000$.}
	\label{Fig:Figure 2}
\end{figure}
\begin{figure}[!htbp]
	\centering
\subfloat{\includegraphics[width=4cm,height=2.5cm]{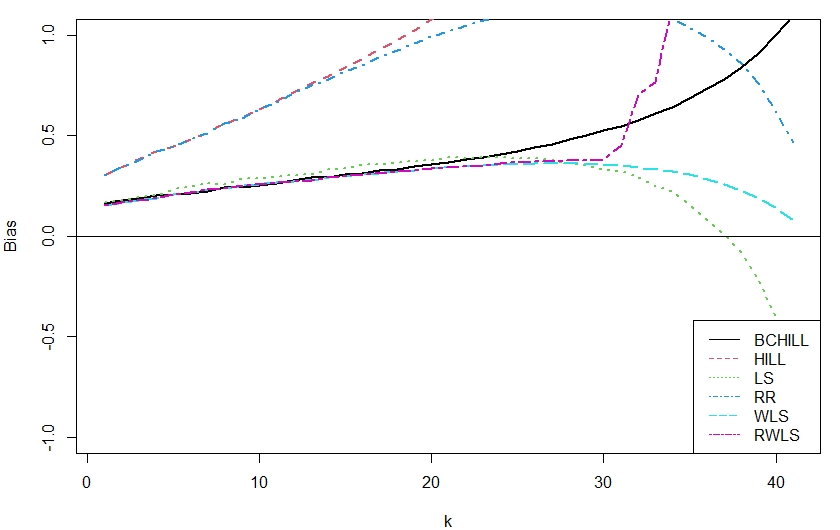}}\hfill
\subfloat{\includegraphics[width=4cm,height=2.5cm]{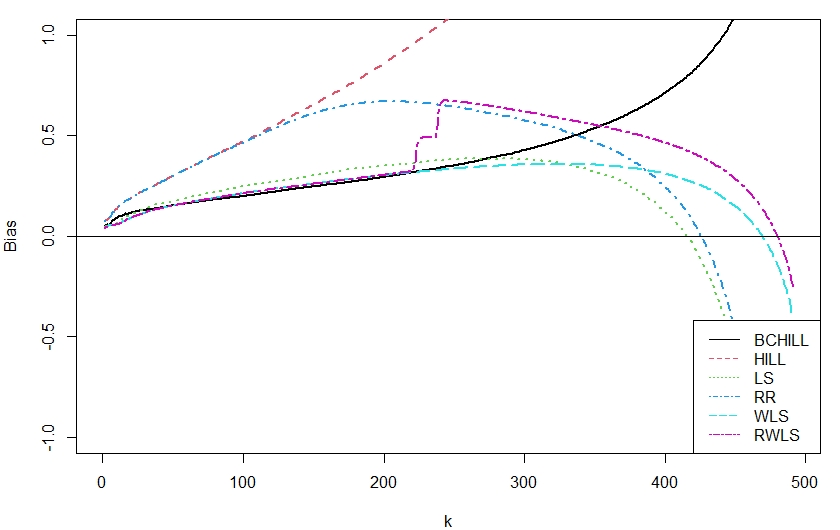}}\hfill
\subfloat{\includegraphics[width=4cm,height=2.5cm]{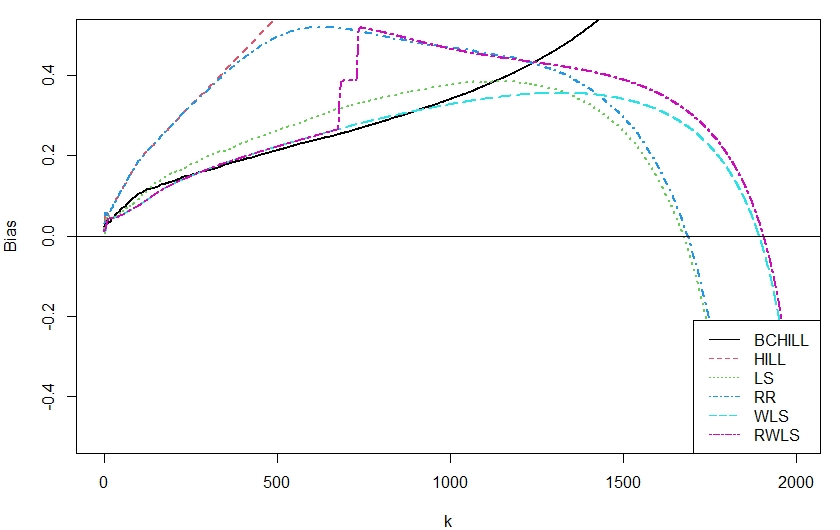}}\hfill
	\\[\smallskipamount]
\subfloat{\includegraphics[width=4cm,height=2.5cm]{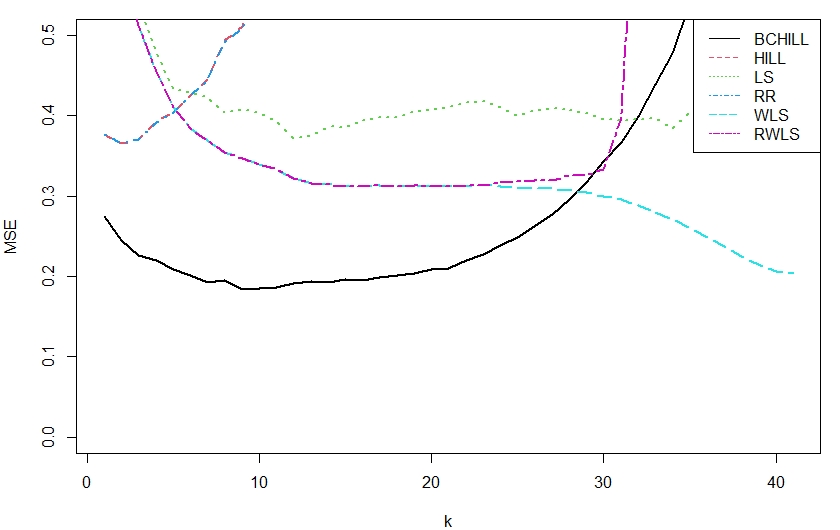}}\hfill
\subfloat{\includegraphics[width=4cm,height=2.5cm]{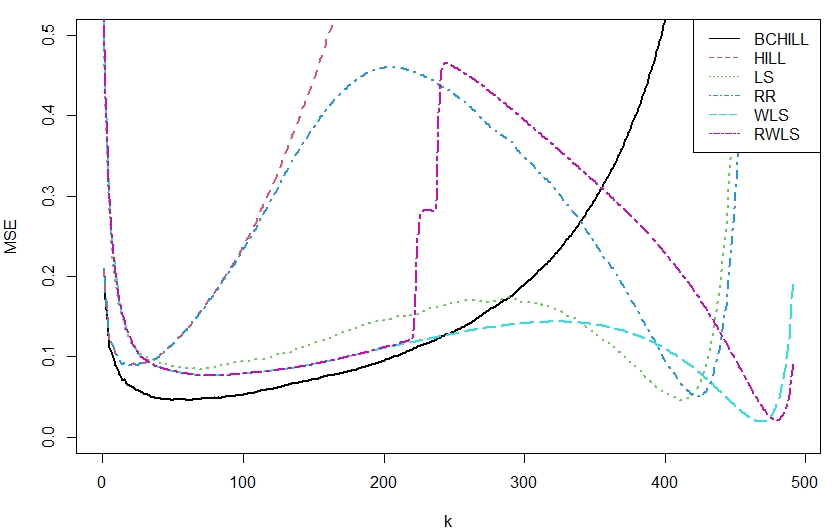}}\hfill
\subfloat{\includegraphics[width=4cm,height=2.5cm]{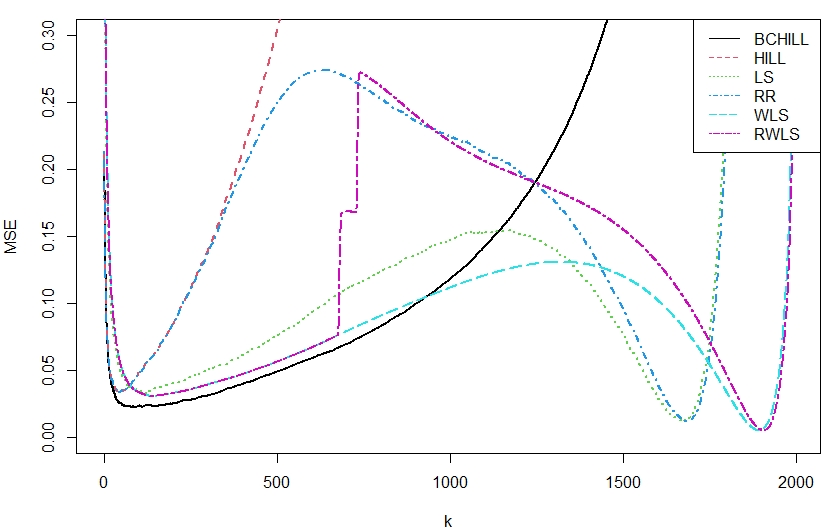}}
	\caption{Results for Burr type XII distribution with $\gamma=1.0$: Bias(top row) and MSE(bottom row). First column: $n=50$; second column: $n=500$; and third column: $n=2000$.}
	\label{Fig:Figure 3}
\end{figure}

\begin{figure}[!htbp]
	\centering
\subfloat{\includegraphics[width=4cm,height=2.5cm]{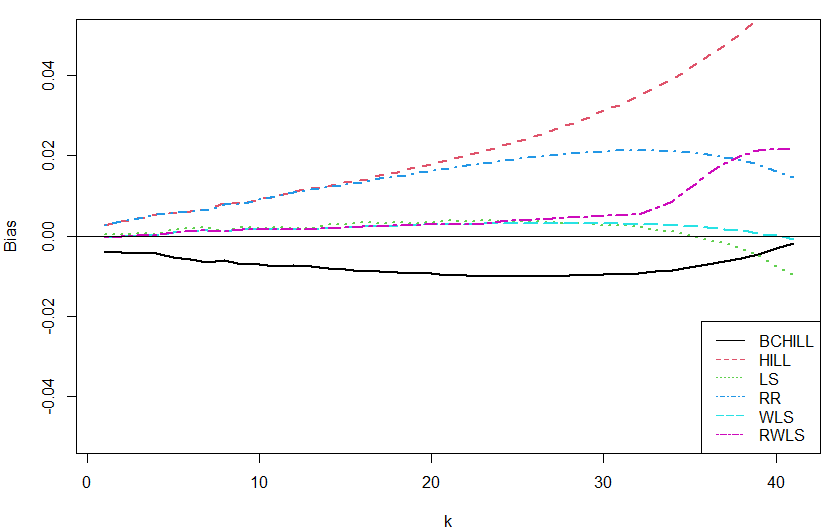}}\hfill
\subfloat{\includegraphics[width=4cm,height=2.5cm]{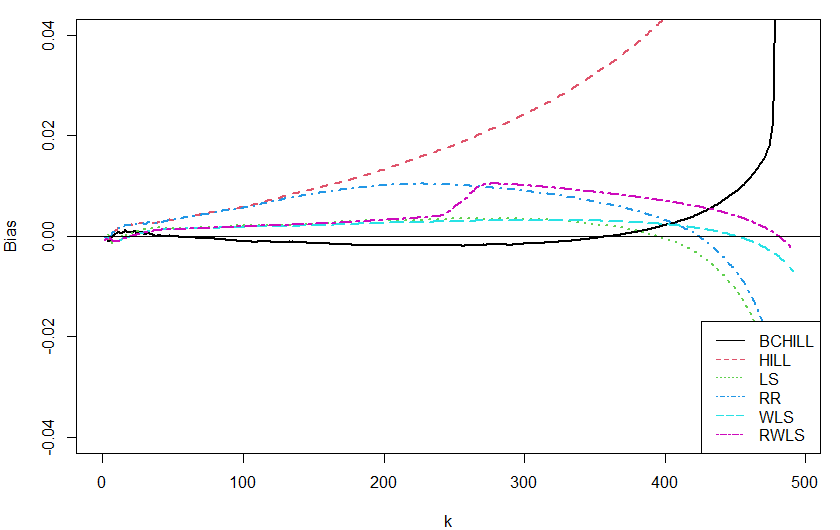}}\hfill
\subfloat{\includegraphics[width=4cm,height=2.5cm]{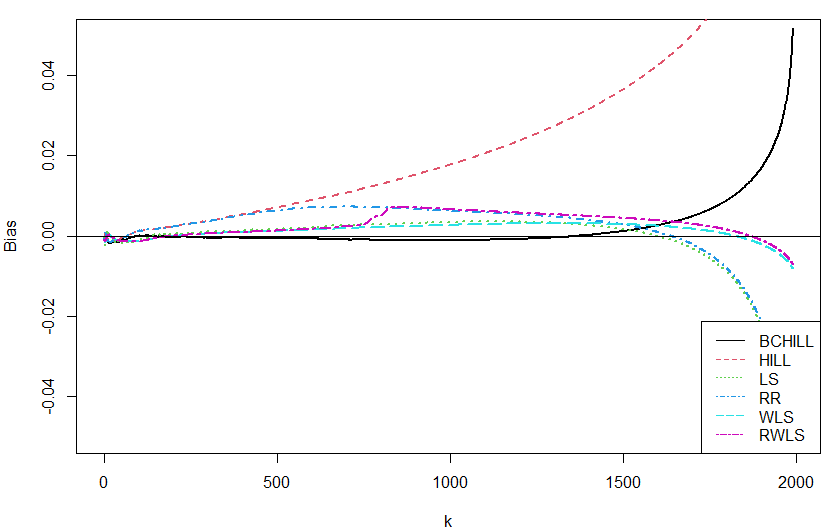}}\hfill
	\\[\smallskipamount]
\subfloat{\includegraphics[width=4cm,height=2.5cm]{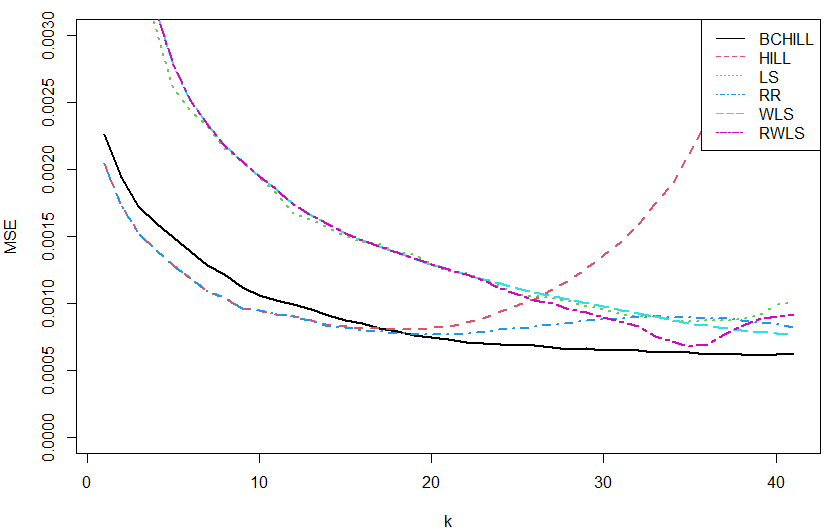}}\hfill
\subfloat{\includegraphics[width=4cm,height=2.5cm]{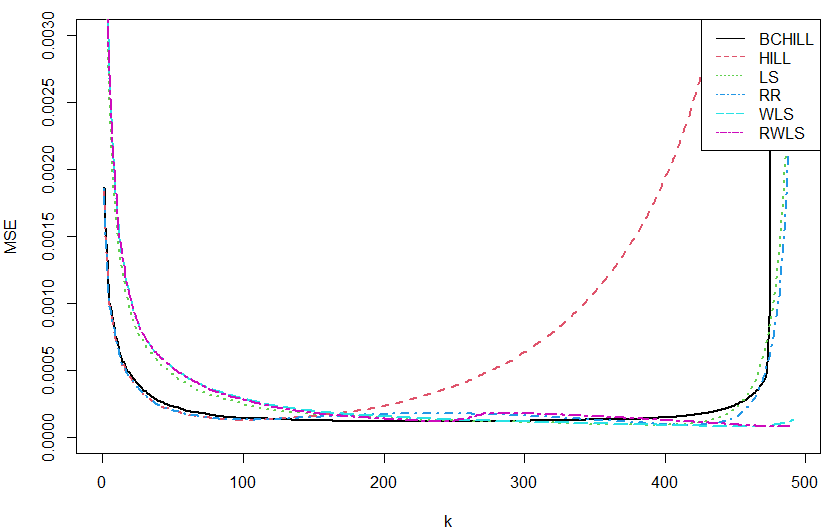}}\hfill
\subfloat{\includegraphics[width=4cm,height=2.5cm]{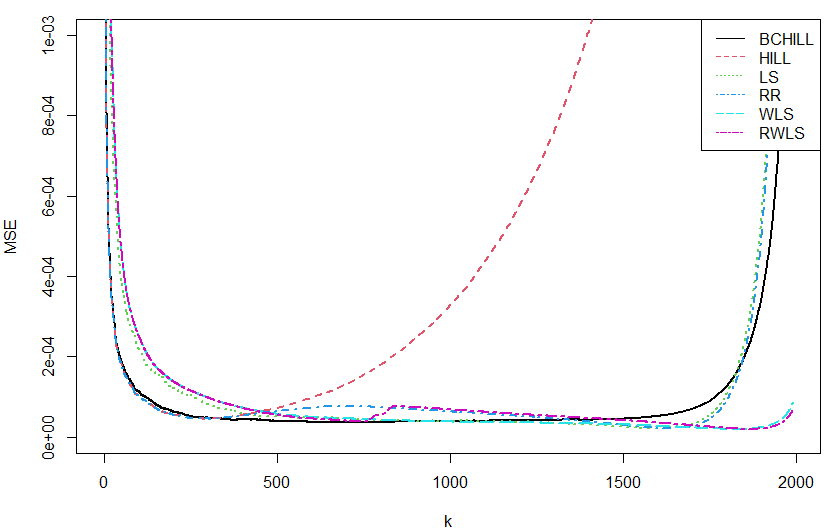}}
	\caption{Results for Fr\'{e}chet distribution with $\gamma=0.1$: Bias(top row) and MSE(bottom row). First column: $n=50$; second column: $n=500$; and third column: $n=2000$.}
	\label{Fig:Figure 4}
\end{figure}
\begin{figure}[!htbp]
	\centering
\subfloat{\includegraphics[width=4cm,height=2.5cm]{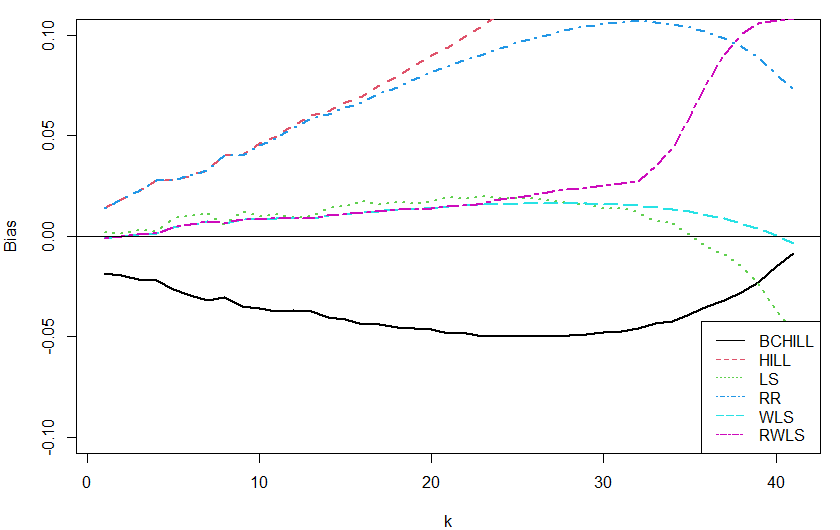}}\hfill
\subfloat{\includegraphics[width=4cm,height=2.5cm]{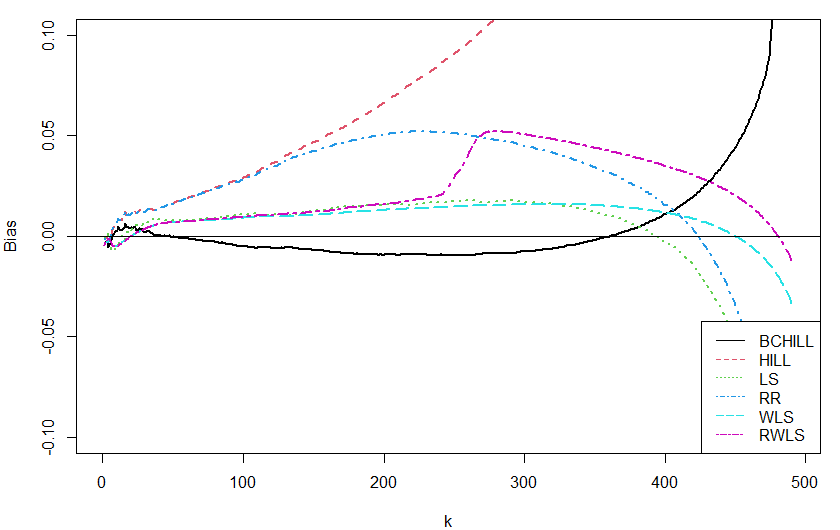}}\hfill
\subfloat{\includegraphics[width=4cm,height=2.5cm]{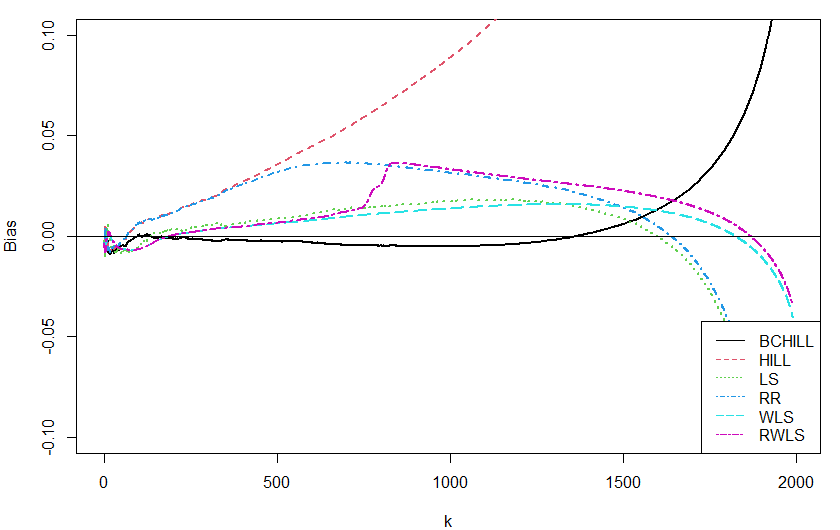}}\hfill
	\\[\smallskipamount]
\subfloat{\includegraphics[width=4cm,height=2.5cm]{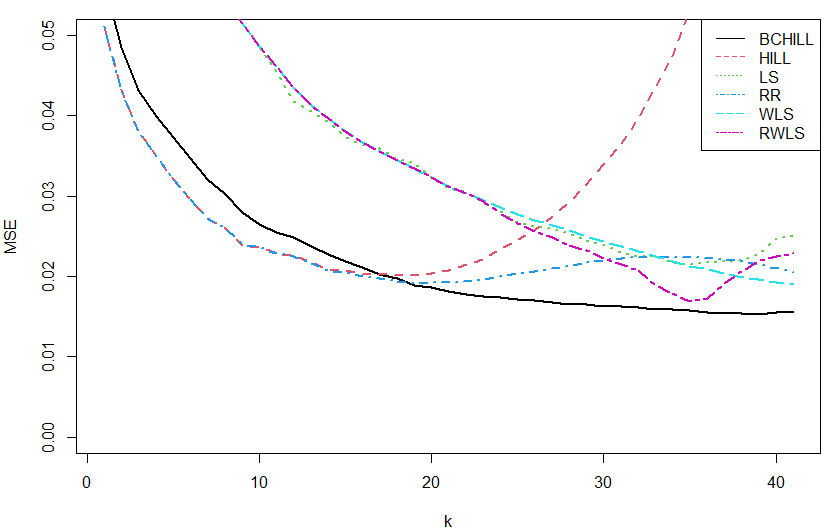}}\hfill
\subfloat{\includegraphics[width=4cm,height=2.5cm]{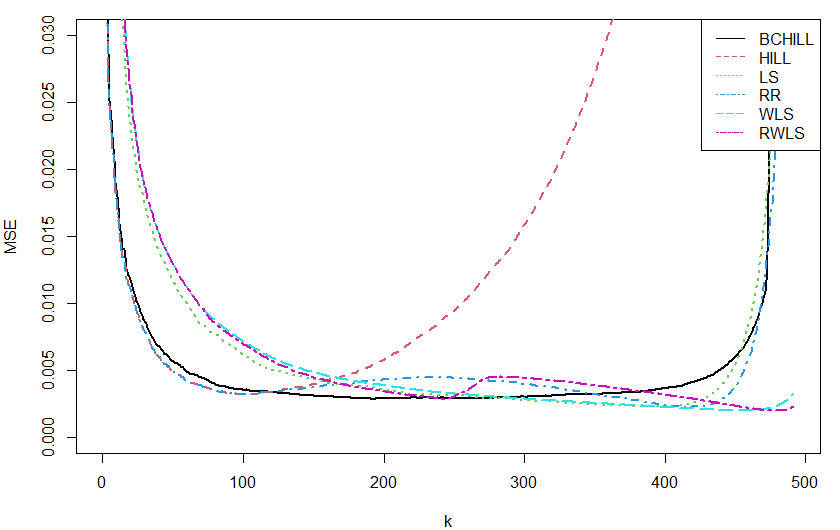}}\hfill
\subfloat{\includegraphics[width=4cm,height=2.5cm]{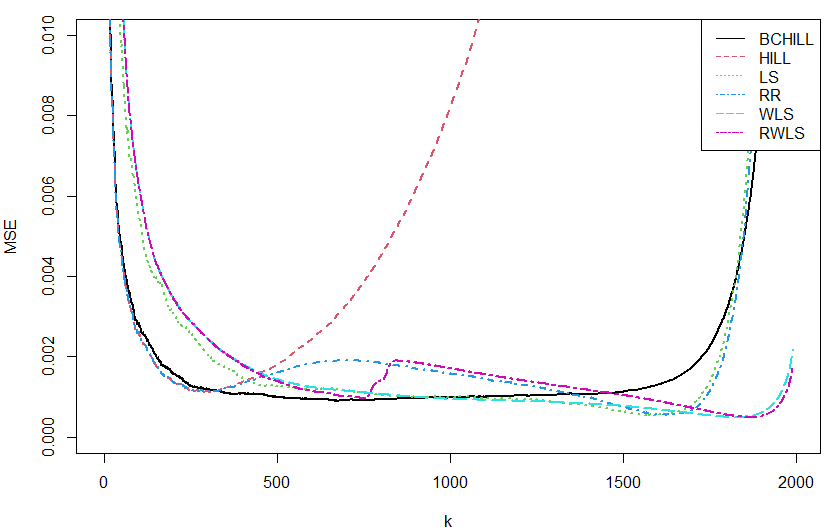}}
	\caption{Results for Fr\'{e}chet distribution with $\gamma=0.5$: Bias(top row) and MSE(bottom row). First column: $n=50$; second column: $n=500$; and third column: $n=2000$.}
	\label{Fig:Figure 5}
\end{figure}

\begin{figure}[!htbp]
	\centering
	\subfloat{\includegraphics[width=4cm,height=2.5cm]{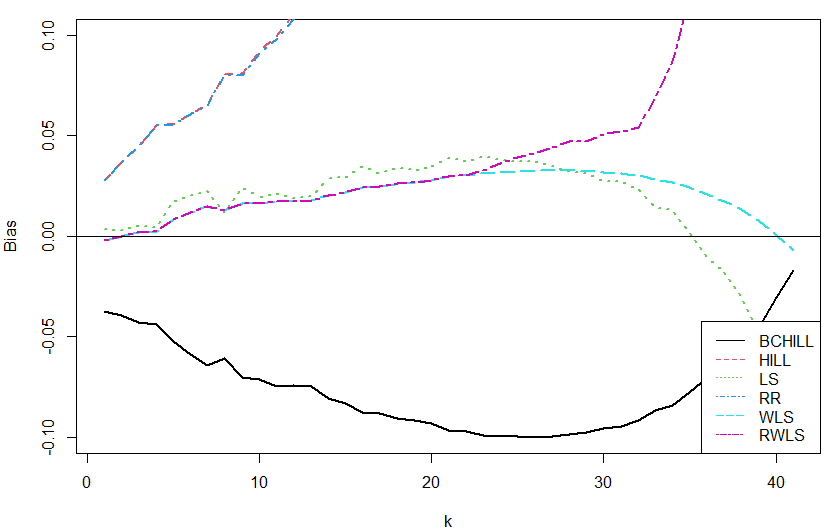}}\hfill
	\subfloat{\includegraphics[width=4cm,height=2.5cm]{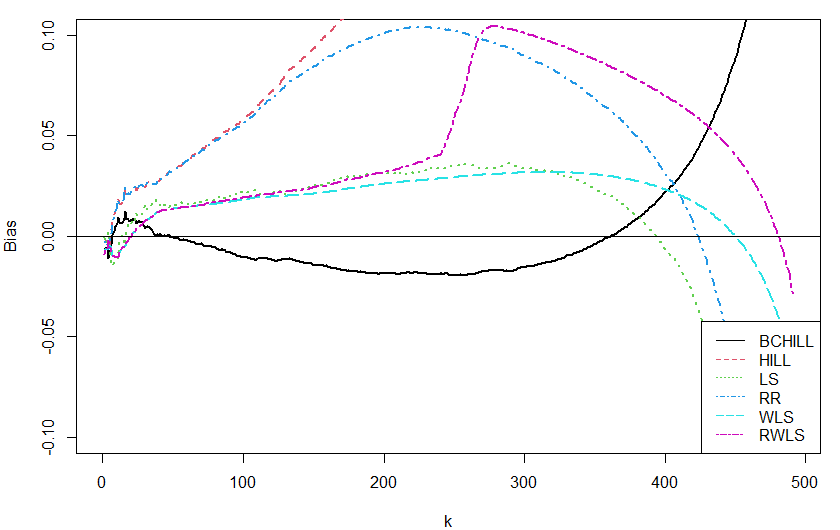}}\hfill
	\subfloat{\includegraphics[width=4cm,height=2.5cm]{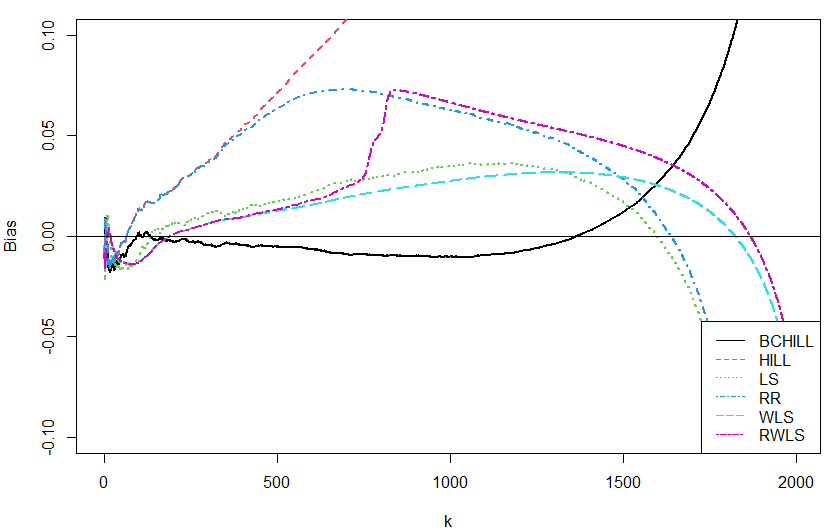}}\hfill
	\\[\smallskipamount]

	\subfloat{\includegraphics[width=4cm,height=2.5cm]{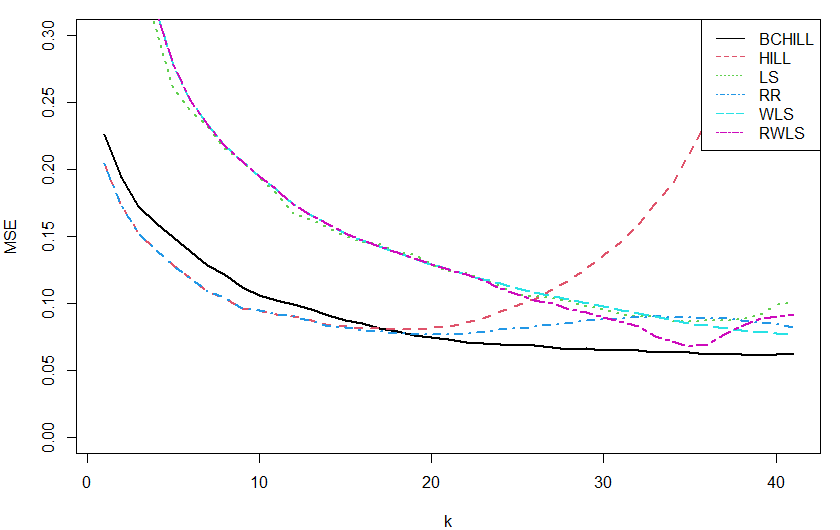}}\hfill
	\subfloat{\includegraphics[width=4cm,height=2.5cm]{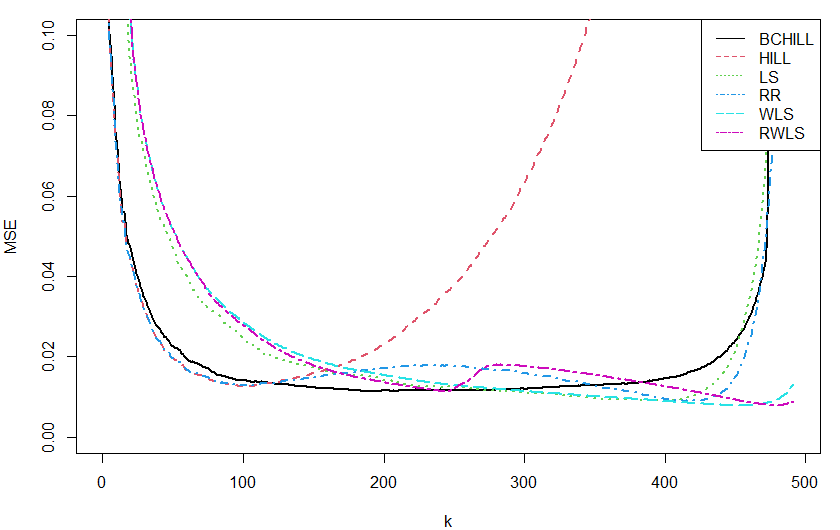}}\hfill
	\subfloat{\includegraphics[width=4cm,height=2.5cm]{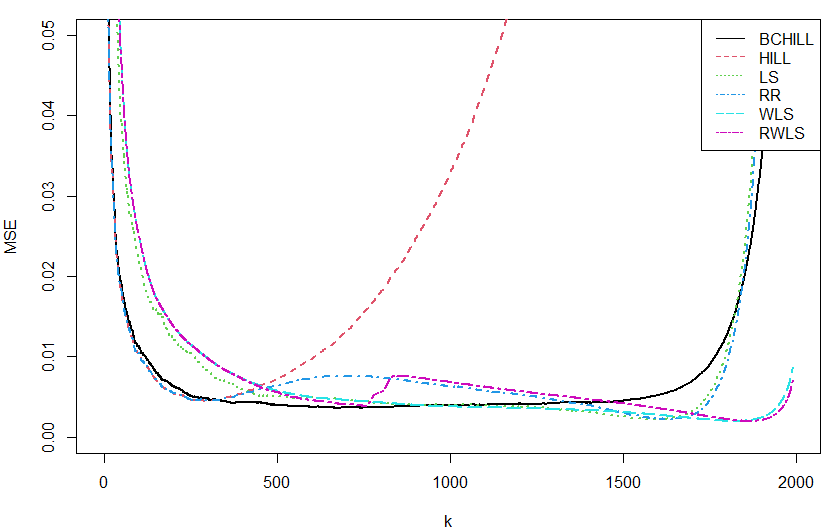}}
	\caption{Results for Fr\'{e}chet distribution with $\gamma=1.0$: Bias(top row) and MSE(bottom row). First column: $n=50$; second column: $n=500$; and third column: $n=2000$.}
	\label{Fig:Figure 6}
\end{figure}

\section{Applications}

In this section, we consider the estimation of the tail index of the underlying distribution of two datasets from the insurance industry. First, 
the SOA Group Medical Insurance dataset consists of over 170,000 claims recorded from 1991 to 1992. In this study, we consider the 1991 dataset, which comprised 75,789 claims and have been studied widely in the extreme value context (see e.g \cite{Beirlant2004},\cite{minkah2021robust}). Considering the large size of this dataset, we focus on the extreme tail of the data and hence consider the top 10\%  data points, (i. e, $0<k\le0.1n$). The SOA dataset is available at \url{https://lstat.kuleuven.be/Wiley/Data/soa.txt}.

Second, the automobile insurance data from Ghana consists of 452 claims from July 7, 2020, to May 11, 2021 and can be found at \url{https://github.com/kikiocran/TailEstimators} We will refer to this dataset as the GH claims in this study. To the best of our knowledge this dataset has never been used in the extreme value theory literature.

The scatter plots of the SOA, and the GH claims are shown in Figure~\ref{Fig:Figure 7}. We observe that two claims and one claim in the SOA and GH claims, respectively, appear to be far detached from the bulk of the data. These observations can also be seen to deviate from linearity and far removed from the bulk of the points respectively in the Pareto and  exponential Q-Q plots (Figure~\ref{Fig:Figure 8}) of the two datasets. Such large observations are suspected outliers and may significantly influence the tail index estimates (see e.g \cite{minkah2021robust}). The convex curvature of the exponential Q-Q plots and the near linearity of the Pareto Q-Q plots of the datasets indicate that the datasets suggest that they belong to the Pareto-type distributions.
  
\begin{figure}[!htbp]
	\centering
	\subfloat{\includegraphics[width=6cm,height=4cm]{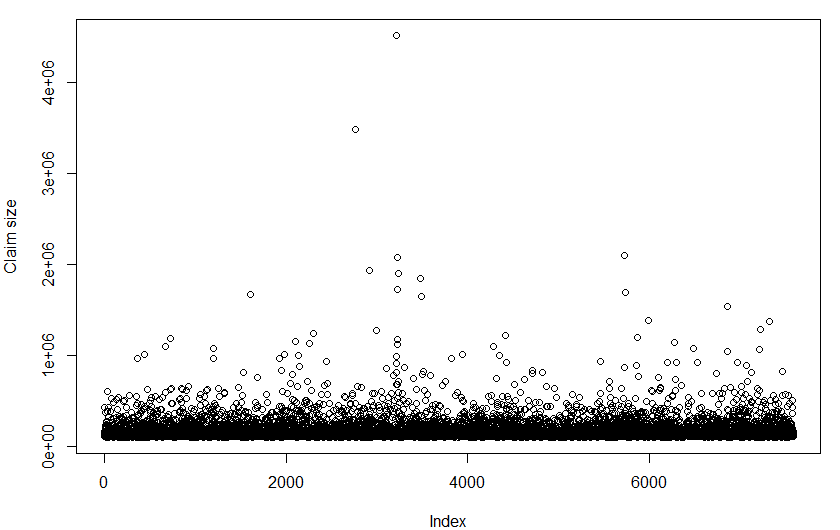}}\hfill
	\subfloat{\includegraphics[width=6cm,height=4cm]{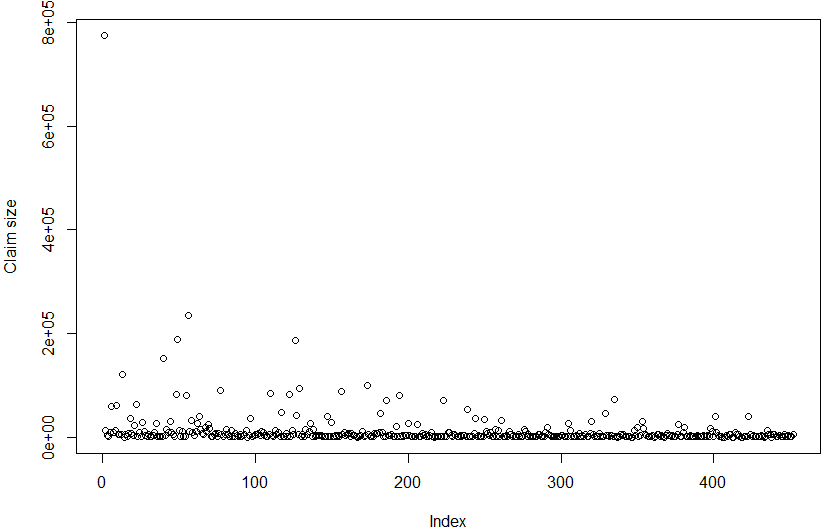}}\hfill
	\caption{Scatter plots: SOA claims data (left panel); GH claims data (right panel).}
	\label{Fig:Figure 7}
\end{figure}

\begin{figure}[!htbp]
	\centering
	\subfloat{\includegraphics[width=6cm,height=4cm]{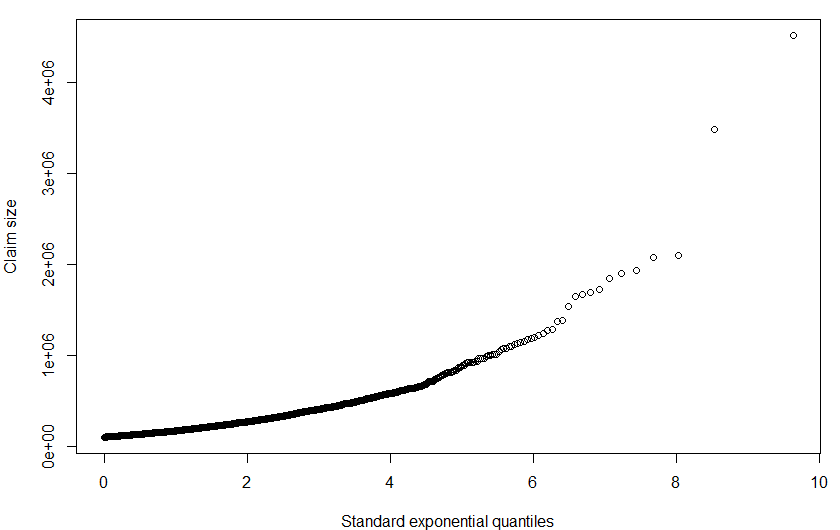}}\hfill
	\subfloat{\includegraphics[width=6cm,height=4cm]{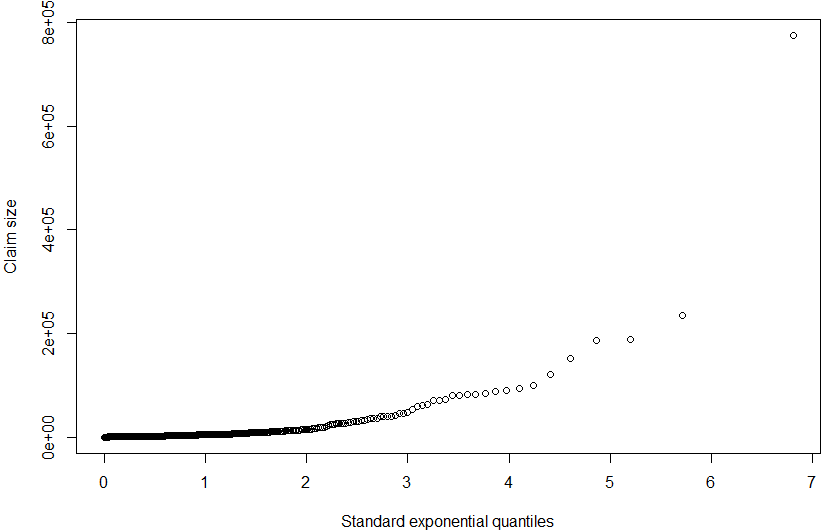}}\hfill
	\\[\smallskipamount]
	\subfloat{\includegraphics[width=6cm,height=4cm]{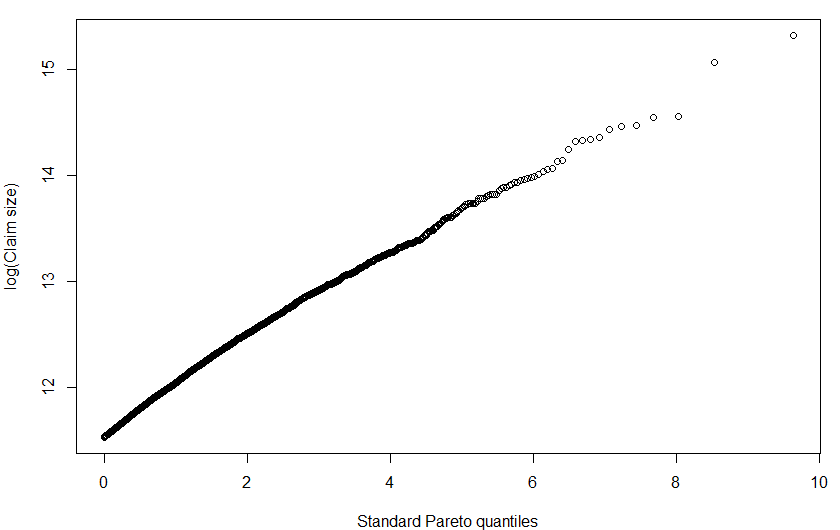}}\hfill
	\subfloat{\includegraphics[width=6cm,height=4cm]{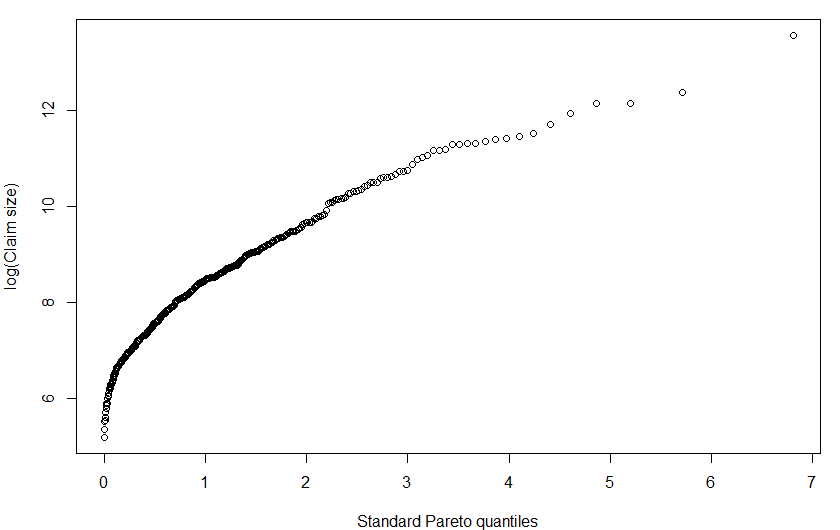}}\hfill
	\caption{Exponential and Pareto Q-Q plots: SOA claims (left panel); GH claims (right panel).}
	\label{Fig:Figure 8}
\end{figure}

 Figure~\ref{Fig:Figure 9} shows the sample paths of the tail index estimators for the underlying distributions of the two datasets. The plot of HILL diverges as $k$ increases, i.e., it is very sensitive to the changes in $k$. Hence, it is not an appropriate estimator for estimating the tail index. The other estimators exhibit some form of stability, however, the sample paths of the proposed estimators (i.e., RWLS and WLS) are smooth, that is, these estimators are less sensitive to changes in $k$. All the tail index estimators considered are very unstable for small values of $k$ due to the small number of exceedances. A specific tail index estimate can be obtained from the plots of WLS and RWLS for both datasets.

\begin{figure}[!htbp]
	\centering
	\subfloat[]{\includegraphics[width=6.2cm,height=5cm]{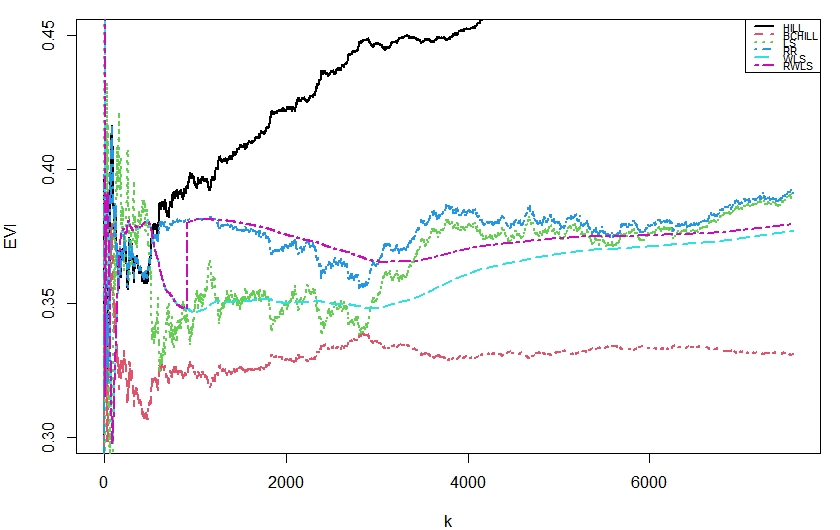}}\hfill
	\subfloat[]{\includegraphics[width=6.2cm,height=5cm]{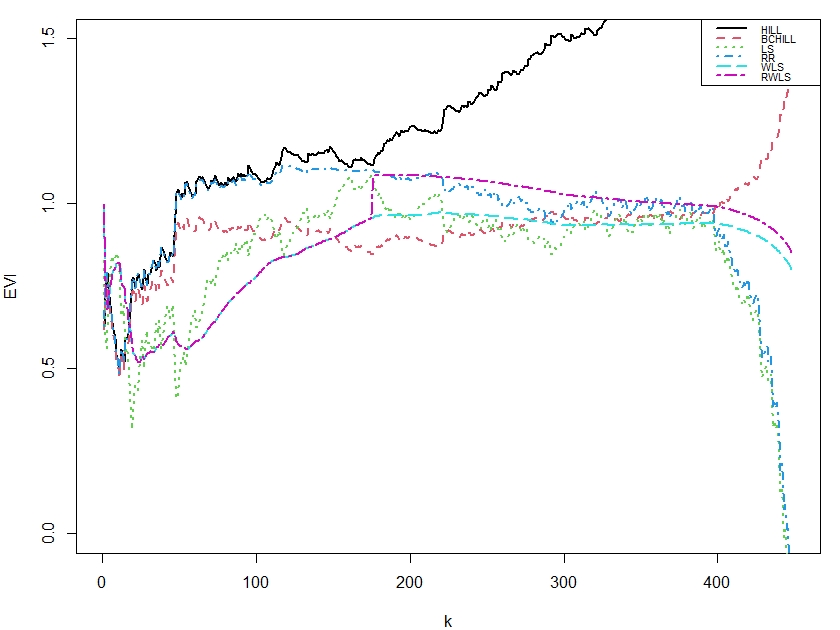}}\hfill

	\caption{Tail index estimates as a function of $k$: (a) SOA data; (b) GH claims.}
	\label{Fig:Figure 9}
\end{figure}

\section{Conclusion}
In this paper, we proposed tail index estimators for the Pareto-type of distributions using the regression model. In addition to the ordinary least squares and the ridge regression estimators, we proposed the regularised weighted least squares and the weighted least squares estimators as alternative reduced-bias estimators. The tail index estimates by the proposed estimators are generally stable and smooth across a broader path of $k$. The characteristics of the proposed estimators are as follows;
\begin{enumerate}
	\item[$\bullet$] they are asymptotically consistent, unbiased and normally distributed with mean $0$ and variance $\gamma^2$.
	\item[$\bullet$] the MSE curves are low and flat over the central part of $k$.
	\item[$\bullet$] the plots of their tail index estimates are more stable, smooth and horizontal than the Hill, ordinary least squares and the bias-corrected Hill estimators.  
\end{enumerate}

In conclusion, comparatively, the proposed estimators are competitive to the existing estimators and can be considered as appropriate estimators of the tail index with MSE, bias and in real-life application. 

\section{Proofs}
\begin{proof}[Proof of  Lemma~\ref{L1}]
	
	\begin{enumerate}

		\item[i.] From Eqn. (\ref{E20})  and Eqn. (\ref{E22}), we have
		\begin{align*}
			S_{1}(\theta)
			&=\sum_{j=1}^{k}\tilde{W}_{j}C_{j}\\
	&=2\left(\dfrac{\alpha(k)+1}{2\alpha(k)+1}\right)\int_{0}^{1}\left(1-\EX(\theta^{\alpha(k)})u\right)u^{-\rho}\,du +O(1)\\
			&=\dfrac{2\left(2\frac{\alpha(k)}{k}-\rho\frac{\alpha(k)}{k}+\frac{1}{k}\right)}{\left(2\frac{\alpha(k)}{k}+\frac{1}{k}\right)(1-\rho)(2-\rho)}+O(1).
		\end{align*}
		
		It follows that,
		\begin{align*}
			\lim\limits_{k\to \infty}\Big[\dfrac{2\left(2\frac{\alpha(k)}{k}-\rho\frac{\alpha(k)}{k}+\frac{1}{k}\right)}{\left(2\frac{\alpha(k)}{k}+\frac{1}{k}\right)(1-\rho)(2-\rho)}+O(1)\Big]
			&=\dfrac{2\left(2\varDelta-\rho\varDelta\right)}{\left(2\varDelta\right)(1-\rho)(2-\rho)}\\
			&=\dfrac{1}{\left(1-\rho\right)},
		\end{align*}
		where $\varDelta=\lim\limits_{k\to \infty}\alpha(k)/k$  and  
			hence, as $k\to \infty$, we  have 
		$S_{1}(\theta)\stackrel{a.s}\longrightarrow 1/(1-\rho).$

		\item[ii.] From Eqn. (\ref{E23}),
		\begin{equation*}
			S_{2}(\theta)=\sum_{j=1}^{k}\tilde{W}_{j}C_{j}^2-\left(\sum_{j=1}^{k}\tilde{W}_{j}C_{j}\right)^2=\sum_{j=1}^{k}\tilde{W}_{j}C_{j}^2-\left(\dfrac{1}{1-\rho}\right)^2.
		\end{equation*}
		Using Eqn. (\ref{E20}), the first term can be written as
		\begin{align*}
			\sum_{j=1}^{k}\tilde{W}_{j}C_{j}^2
			&=2\left(\dfrac{\alpha(k)+1}{2\alpha(k)+1}\right)\int_{0}^{1}\left(1-\EX(\theta^{\alpha(k)})u\right)u^{-2\rho}\,du+O(1)\\
			&=\dfrac{\left(2\frac{\alpha(k)}{k}-2\rho\frac{\alpha(k)}{k}+\frac{1}{k}\right)}{\left(2\frac{\alpha(k)}{k}+\frac{1}{k}\right)(1-\rho)(1-2\rho)}+O(1).
		\end{align*}
		
		It follows that,
		\begin{align*}
			\lim\limits_{k\to \infty}\Big[\dfrac{\left(2\frac{\alpha(k)}{k}-2\rho\frac{\alpha(k)}{k}+\frac{1}{k}\right)}{\left(2\frac{\alpha(k)}{k}+\frac{1}{k}\right)(1-\rho)(1-2\rho)}+O(1)\Big]
			&=\dfrac{\left(2\varDelta-2\rho\varDelta\right)}{2\varDelta(1-\rho)(1-2\rho)}=\dfrac{1}{\left(1-2\rho\right)},
		\end{align*}
		
		where $\varDelta=\lim\limits_{k\to \infty}\alpha(k)/k$. Therefore, 
		$$S_{2}(\theta)\longrightarrow \dfrac{1}{1-2\rho}-\left(\dfrac{1}{1-\rho}\right)^2=\dfrac{\rho^2}{\left(1-2\rho\right)\left(1-\rho\right)^2}$$ as $k\to \infty$. That is, $S_{2}\left(\theta\right)\stackrel{a.s}\longrightarrow \rho^2/(1-2\rho)(1-\rho)^2$ as $k\to \infty$. 
		
		\item[iii.] The expression $\dot{S}(\theta)=\sum_{j=1}^{k}\tilde{W}_{j}^2\left(S_{1}(\theta)-C_{j}\right)$ can also be written as
		\begin{align*}
			\dot{S}(\theta)
			&=\dfrac{4}{k}\left(\dfrac{\alpha(k)+1}{2\alpha(k)+1}\right)^2\left\{\dfrac{1}{k}\sum_{j=1}^{k}\left(1-\theta_{j}^{\alpha(k)}\dfrac{j}{k+1}\right)^2\left(S_{1}(\theta)-\left(\dfrac{j}{k+1}\right)^{-\rho}\right)\right\}\\
			&=\dfrac{4}{k}\left(\dfrac{\frac{\alpha(k)}{k}+\frac{1}{k}}{2\frac{\alpha(k)}{k}+\frac{1}{k}}\right)^2\int_{0}^{1}\left(1-\EX(\theta^{\alpha(k)})u\right)^2\left(\dfrac{1}{1-\rho}-u^{-\rho}\right)\,du+O(1).
		\end{align*}
		Therefore, as $k\to \infty$, $\dot{S}(\theta)\stackrel{a.s}\longrightarrow 0$.
		
		\item[iv.] $\ddot{S}(\theta)=\sum_{j=1}^{k}\tilde{W}_{j}^2\left(S_{1}(\theta)-C_{j}\right)^2$ can also be expressed as
		\begin{align*}
			\ddot{S}(\theta)
			&=\dfrac{4}{k}\left(\dfrac{\alpha(k)+1}{2\alpha(k)+1}\right)^2\left\{\dfrac{1}{k}\sum_{j=1}^{k}\left(1-\theta_{j}^{\alpha(k)}\dfrac{j}{k+1}\right)^2\left(S_{1}(\theta)-\left(\dfrac{j}{k+1}\right)^{-\rho}\right)^2\right\}\\
			&=\dfrac{4}{k}\left(\dfrac{\frac{\alpha(k)}{k}+\frac{1}{k}}{2\frac{\alpha(k)}{k}+\frac{1}{k}}\right)^2\int_{0}^{1}\left(1-\EX(\theta^{\alpha(k)})u\right)^2\left(\dfrac{1}{1-\rho}-u^{-\rho}\right)^2\,du+O(1).
		\end{align*}
		It also follows that, as $k\to \infty$,  $\ddot{S}(\theta)\stackrel{a.s}\longrightarrow 0$.
		
	\end{enumerate}
\end{proof}

\begin{proof}[Proof of  Lemma~\ref{L2}]
	The proof of Lemma~\ref{L2} easily follows by using Lemma~\ref{L1}.
\end{proof}

\begin{proof}[Proof of Lemma~\ref{L3}]  We  observe that  $$\dfrac{-k^\rho}{|k^{1+\omega+\rho}+n^{\rho}|} \le \mathfrak{M}_{\lambda}(\theta,C_j)\le \dfrac{k^{\rho}C_{j}^2/S_{2}(\theta)}{|k^{1+\omega+\rho}+n^{\rho}|} .$$
	Therefore, we  have  $\mathfrak{M}_{\lambda}(\theta, C_j) \to O(1/k^{1+\omega}),\hspace{4pt}0<\omega\le 0.1$, as $k\to \infty$,  which  completes  the  proof  of  Lemma~\ref{L3}.
\end{proof}

\begin{proof}[Proof of  Lemma~\ref{L4}]
	The proof requires the use of large deviation principles (LDP). From Eqn.(\ref{E15}) and Eqn.(\ref{E34}),
	$\hat b_{n,k}(\theta,\hat\rho)=\sum_{j=1}^{k}\mathfrak{M}_{\lambda}(\theta,C_j) T_{j}.$  	Given  $\{\hat\rho=\rho\} $, $T_{j}$ is exponentially distributed with mean $$\mu_{j}(\theta,\rho)=\EX[T_j|\hat\rho=\rho,\theta]=\gamma+\hat{b}_{n,k}(\hat\rho,\theta)C_j(\hat\rho).$$ Therefore,  we  have
	
	\begin{equation}\label{E36}
		M_{T_{k},\theta,\rho}(t)=\EX[e^{tT_k}|\theta,\hat\rho=\rho]
		=\Big(\dfrac{1}{1-\mu_{k}(\theta,\rho)t}\Big).
	\end{equation}
	
	Using Lemma~\ref{L4},  Eqn.(\ref{E36})  and  similar  calculations  as  in  \cite{ocran2021reduced}, the moment generating function of  $b_{n,k}$ given  the  law  of  $\{\theta,\hat\rho=\rho\} $  is  
	
		\begin{equation*} M_{\hat{b}_{n,k}(\rho)}(t)=
\left\{\prod_{j=1}^{k}\left(\dfrac{1}{1-\mu_{j}(\theta,\rho)\mathfrak{M}_{\lambda}(\theta)t}\right)\right\}.
\end{equation*}

	It follows that

	\begin{equation*}
		M_{\hat b_{n,k}}(kt)=\left\{\prod_{j=1}^{k}\left(\dfrac{1}{1-k\mu_{j}\mathfrak{M}_{\lambda}(\theta, C_j)t}\right)\right\}.
	\end{equation*}

	Now  using the  bound  on  $\mathfrak{M}_{\lambda}(\theta,C_j)$, see Lemma~\ref{L3}  and     the Squeeze Theorem, we obtain
	\begin{equation*}
		\lim\limits_{k \to \infty}\dfrac{1}{k}\log M_{b_{n,k}}(kt)=0.
	\end{equation*}
	Hence, by the G$\ddot{a}$rtner Ellis Theorem, conditional  on $\{\theta,\rho\},$ the  statistics  $\hat b_{n,k}(\theta,\rho)$ follows a Large Deviation Principle (LDP) with speed $k$ and a rate function $I(x)$ defined as
	\begin{align*}
		I(x)
		&=\sup_{\eta\in\mathbb{R}}\left\{\eta x-0\right\}
		=\sup_{\eta\in\mathbb{R}}\left\{\eta x\right\}\\
		&=\begin{cases}
			0 \hspace{6pt} \text{if} \hspace{6pt} x=0,
			\\
			\infty \hspace{6pt} \text{if} \hspace{6pt} x\ne 0,
		\end{cases}
	\end{align*}
	where $x\in [0, \infty)$. This  implies  for every $ \epsilon>0$, we  have  
	\begin{align*}
		\prob\left(\sqrt{k}b_{k,n}>\epsilon\right)
		&=\prob\left(b_{n,k}>\frac{\epsilon}{\sqrt{k}} \Big|(\hat\theta,\hat\rho)=(\theta,\rho)\right)\prob\Big[(\hat\theta,\hat\rho)=(\theta,\rho)\Big]\\
		&=\prob\left(b_{n,k}>\varGamma_{k}\right)\prob\Big[(\hat\theta,\hat\rho)=(\theta,\rho)\Big]
		\le e^{-kI(x)+o(k)},
	\end{align*} 
	
	where $\varGamma_{k}=\epsilon/k$. The atypical behaviour of the rate function is when $x\ne 0$, therefore
	\begin{equation*}
		\lim\limits_{k \to \infty}\prob\left(\sqrt{k}\hat b_{n,k}>\epsilon\right)\le 0.
	\end{equation*}
	Thus, $\sqrt{k}\hat b_{n,k}\rightarrow_{p} 0$ as $k \to \infty$ and this ends the proof.
\end{proof}
	
	\begin{proof}[Proof of Lemma~\ref{L5}]
		We  observe  that  $T_{1}, T_{2}, T_{3},...$ are independent but not identical  distributed  random  variables.
		\begin{itemize}
			\item [(i)] $$\begin{aligned}  \EX(T_j)=\EX\Big\{\EX\left(T_{j}\big|(\theta, \hat \rho\right)\Big\}&= \EX\Big\{\gamma+\hat b_{n,k}(\theta,\hat\rho)C_j(\hat\rho)\Big\}\le\gamma+\EX(\hat b_{n,k}(\theta,\hat\rho))\\
			&\le \gamma+ \EX\Big[\frac{(1-2\hat\rho){(1-\hat\rho)^2}k^{\hat\rho+1}\sum_{j=1}^{k}C_j(\hat\rho)Z_j/k}{ \hat\rho^2|k^{1+\omega+\hat\rho}+n^{\hat\rho}|}+o(1)\Big]\\
			&\le\gamma+\EX\Big[\frac{(1-2\hat\rho){(1-\hat\rho)^2}}{\hat\rho^2k^{\omega}} \EX\Big[\sum_{j=1}^{k}C_j(\hat\rho)Z_j/k \Big]\Big]+o(1)\,\,\, \mbox{a.s.}\\
			&\le\gamma+\EX\Big[\frac{(1-2\hat\rho){(1-\hat\rho)^2}}{\hat \rho^2}\Big] \Big[\frac{\EX(\gamma_k^{H})}{k^{\omega}}\Big] +o(1)\,\,\, \mbox{a.s.}\\
			&\le\gamma+ o(1/k^{\omega})+o(1), 	
			\end{aligned}$$
			
			where  $0<\omega<0.1$    and  $\gamma_k^{H}$  is  the  Hill  estimator.   	Hence  we  have  $$\mu_j=\EX\Big[\EX_{\theta}(\tilde{W}(\theta_j))T_j\Big]\le\EX_{\theta}[\tilde{W}_j(\theta_j)]( \gamma+ o(1/k^{\omega})+o(1)).  $$
			
			\item[(ii)] Let  $f$  be  the  probability  density  function  of  $ T_j  $   given $(\theta,\hat\rho)$  and  observe  that  we
			
			\begin{align*}
			\EX\left(|Z_{j}-\mu_{j}|^{2+\delta}\big|\theta, \hat \rho=\rho\right)
			&= \int_{0}^{\infty}|Z_{j}-\mu_{j}|^{2+\delta}f(z_{j})\,d{z}_{j}\\
			&=-\int_{0}^{\mu_{j}}\left(Z_{j}-\mu_{j}\right)^{2+\delta}f(z_{j})\,d{z}_{j}+\int_{\mu_{j}}^{\infty}\left(Z_{j}-\mu_{j}\right)^{2+\delta}f(z_{j})\,d{z}_{j}\\
			&=\dfrac{e^{-1}}{\mu_{j}}\left\{\dfrac{(-\mu_{j})^{3+\delta}(7+2\delta)}{(3+\delta)(4+\delta)}+\mu_{j}^{3+\delta}(2+\delta)!\right\}\{1+O(1)\}\\
			&=\mu_{j}^{2+\delta}\eta(\delta)\{1+O(1)\},
			\end{align*}
			where 
			$\eta(\delta)=\dfrac{1}{e}\left\{\dfrac{(-1)^{3+\delta}(7+2\delta)}{(3+\delta)(4+\delta)}+(2+\delta)!\right\}<\infty$ for  $0<\delta<\infty$.  Therefore  we  have  that   $$\displaystyle \EX\left(|Z_{j}-\mu_{j}|^{2+\delta}\right)=\EX\Big\{\EX\left(|T_{j}-\mu_{j}|^{2+\delta}\big| (\theta, \hat \rho)\right)\Big\}\le \eta(\delta)\{1+O(1)\} \mu_{j}^{2+\delta}.  $$

			Now define
			
			$Z_{j}=\EX_{\theta}[\tilde{W}_j(\theta_j)]T_{j}-\mu_{j}, 1\le j \le k,$   and  note  that

			$$\begin{aligned}
			S_{n}^2 &=Var\left(\sum_{j=1}^{k}\EX_{\theta}Z_{j}\right)\\
			&=\sum_{j=1}^{k}Var(\EX_{\theta}\tilde{W}_j(\theta_j)T_{j})=\sum_{j=1}^{k}[\EX_{\theta}\tilde{W}_j(\theta)]^2Var (T_{j})=\gamma^2 \sum_{j=1}^{k}[\EX_{\theta}\tilde{W}_j(\theta)]^2
			\end{aligned}$$
	
			Hence,
			$$\begin{aligned}
			&\lim_{k\to \infty}\dfrac{1}{S_{n}^{2+\delta}}\sum_{j=1}^{k}\EX\left(|Z_{j}-\mu_{j}|^{2+\delta}\right)\\
			&\le\lim_{k \to\infty}\Big[\dfrac{\eta(\delta)\{1+o(1)\}\Big(\gamma+ o(1/k^{\omega})+o(1)\Big)^{2+\delta}\sum_{j=1}^{k}[\EX_{\theta}\tilde{W}_j(\theta_j)]^{2+\delta}}
			{(\gamma^2 \sum_{j=1}^{k}\big[\EX_{\theta}\tilde{W}_j(\theta)\big]^2)^{1+\delta/2}}\Big]\\
			&=\eta(\delta)\lim_{k \to\infty}\Big[\dfrac{\sum_{j=1}^{k}[\EX_{\theta}\tilde{W}_j(\theta_j)]^{2+\delta}}
			{(\sum_{j=1}^{k}\big[\EX_{\theta}\tilde{W}_j(\theta)\big]^2)^{1+\delta/2}}\Big]\\
			&\le\eta(\delta)\lim_{k\to\infty}\Big[\dfrac{2^{1+\delta/2}k^{-\delta/2}\kappa(\alpha(k))^{1+\delta/2}(1+o(k))}{(1-o(k))^{1+\delta/2}}\Big]\\
		&	=0
			\end{aligned}$$
			
		\end{itemize}	
	\end{proof}

\begin{proof}[Proof of Theorem~\ref{TH1}]
Using Lemma~\ref{L4} and Lemma~\ref{L5}, we can prove Theorem~\ref{TH1}. It has been established in Lemma~\ref{L4} that $\sqrt{k}\hat b_{n,k} \to_p 0$ as $k \to \infty$. Lemma~\ref{L5} also establishes that the Lyapunov's condition holds for Central Limit Theorem; hence, by the Lyapunov's Central Limit Theorem;
$$\sqrt{k} \left(\EX_{\theta}(\hat{\gamma}_{RW}(\lambda,\theta))-\gamma\right)\longrightarrow_{d} N(\mu, \sigma^2), \hspace{6pt} \text{as}  \hspace{6pt} k\to \infty$$.

Therefore, all we need to complete the proof of Theorem~\ref{TH1}, is to specify the parameters of the normal distribution.

Recall from (\ref{E16}) that
\begin{align*}
\hat{\gamma}_{k}(\lambda):=	\EX_{\theta}(\hat{\gamma}_{RW}(\lambda,\theta))=\EX_{\theta}\Big[\sum_{j=1}^{k}\tilde{W}_{j}T_{j}-\hat b_{n,k}\sum_{j=1}^{k}\tilde{W}_{j}C_{j}\Big].
\end{align*}

Also, from Lemma~\ref{L4}, asymptotically, the second term on the right hand side vanishes, so we would concentrate on the first term of the expression only.

Let
$$S_{k}
	=\sqrt{k}\left(\hat{\gamma}_{k}(\lambda)-\gamma\right)
=\sqrt{k}\left\{\EX_{\theta}\Big[\sum_{j=1}^{k}\tilde{W}_{j}T_{j}-\hat b_{n,k}\sum_{j=1}^{k}\tilde{W}_{j}C_{j}\Big]-\gamma\right\}. $$

The expected value of $S_{k}$ is given as
\begin{align*}
	\EX(S_{k})
	&=\EX_{\theta}\left\{\sqrt{k}\sum_{j=1}^{k}\tilde{W}_{j}\EX(T_{j})-\sqrt{k}\hat b_{n,k}\sum_{j=1}^{k}\tilde{W}_{j}C_{j}-\sqrt{k}\gamma\right\}\\	&=\EX_{\theta}\EX_{\rho}\left\{\sqrt{k}\sum_{j=1}^{k}\tilde{W}_{j}\left\{\gamma+\hat  b_{n,k}(\rho)C_{j}(\rho)\right\}-\sqrt{k}\hat b_{n,k}(\rho)\sum_{j=1}^{k}\tilde{W}_{j}C_{j}(\rho)-\sqrt{k}\gamma\right\}\\
	&=\EX_{\theta}\EX_{\rho}\left\{0\right\}\\
	&=0.
\end{align*}

Hence $\EX(S_k)\to \mu=0$ as $k\to \infty$. Recall that, $k\to \infty$, $\sqrt{k}\hat b_{n,k}\to_p 0,$ $2\kappa(\alpha(k))\to 1$ and   by  assumption  $\left(1-\theta_{j}^{\alpha(k)}\frac{j}{k+1}\right)$ and $T_{j}$ are independent,  therefore,  the variance of $S_{k}$ is given by 
\begin{align*}
	Var(S_{k})
	&=kVar\left(\EX_{\theta}\sum_{j=1}^{k}\tilde{W}_{j}T_{j}\right)\\
	&=\dfrac{4}{k}\kappa^{2}(\alpha(k))\sum_{j=1}^{k}Var\left\{\EX_{\theta}\left(1-\theta_{j}^{\alpha(k)}\dfrac{j}{k+1}\right)T_{j}\right\}\\
	&=\dfrac{4}{k}\kappa^{2}(\alpha(k))\sum_{j=1}^{k}\left\{\EX_{\theta}\left(1-\theta_{j}^{\alpha(k)}\dfrac{j}{k+1}\right)\right\}^2Var (T_{j})\\
	&=\dfrac{4}{k}\kappa^{2}(\alpha(k))\sum_{j=1}^{k}\left\{\left(1-\dfrac{1}{\alpha(k)+1}\dfrac{j}{k+1}\right)\right\}^2Var (T_{j})\\
		&=4\kappa^{2}(\alpha(k))\gamma^2\left\{\left(1-\dfrac{1}{(\alpha(k)+1)}+ o(k)/k\right)\right\}\\
\end{align*}

Using  the  assumption $\alpha(k)/k\to\varDelta<\infty$,  as $ k\to\infty$  we  have	$Var(S_{k}) \longrightarrow \sigma^2=\gamma^2,$  which  completes the proof of Theorem~\ref{TH1}.
\end{proof}


\bibliography{Estimations,Paper1,Cit}

\begin{thebibliography}{10}
\expandafter\ifx\csname url\endcsname\relax
  \def\url#1{\texttt{#1}}\fi
\expandafter\ifx\csname urlprefix\endcsname\relax\def\urlprefix{URL }\fi
\expandafter\ifx\csname href\endcsname\relax
  \def\href#1#2{#2} \def\path#1{#1}\fi

\bibitem{longin2016extreme}
F.~Longin, Extreme events in finance: A handbook of extreme value theory and
  its applications, John Wiley \& Sons, 2016.

\bibitem{kithinji2021adjusted}
M.~M. Kithinji, P.~N. Mwita, A.~O. Kube, Adjusted extreme conditional quantile
  autoregression with application to risk measurement, Journal of Probability
  and Statistics 2021 (2021).

\bibitem{gkillas2018application}
K.~Gkillas, P.~Katsiampa, An application of extreme value theory to
  cryptocurrencies, Economics Letters 164 (2018) 109--111.

\bibitem{minkah2021robust}
R.~Minkah, T.~de~Wet, A.~Ghosh, Robust estimation of pareto-type tail index
  through an exponential regression model, Communications in Statistics-Theory
  and Methods (2021) 1--19.

\bibitem{minkah2020tail}
R.~Minkah, Tail index estimation of the generalised pareto distribution using a
  pivot from a transformed pareto distribution, Science and Development Journal
  4~(1) (2020) 1--19.

\bibitem{rohrbeck2018extreme}
C.~Rohrbeck, E.~F. Eastoe, A.~Frigessi, J.~A. Tawn, Extreme value modelling of
  water-related insurance claims, The Annals of Applied Statistics 12~(1)
  (2018) 246--282.

\bibitem{mcneil1997estimating}
A.~J. McNeil, Estimating the tails of loss severity distributions using extreme
  value theory, ASTIN Bulletin: The Journal of the IAA 27~(1) (1997) 117--137.

\bibitem{magnou2017application}
G.~Magnou, An application of extreme value theory for measuring financial risk
  in the uruguayan pension fund, Compendium: Cuadernos de Econom{\'\i}a y
  Administraci{\'o}n 4~(7) (2017) 1--19.

\bibitem{afuecheta2020application}
E.~Afuecheta, C.~Utazi, E.~Ranganai, C.~Nnanatu, An application of extreme
  value theory for measuring financial risk in brics economies, Annals of Data
  Science (2020) 1--40.

\bibitem{mehrnia2021wireless}
N.~Mehrnia, S.~Coleri, Wireless channel modeling based on extreme value theory
  for ultra-reliable communications, IEEE Transactions on Wireless
  Communications (2021).

\bibitem{finkenstadt2003extreme}
B.~Finkenstadt, H.~Rootz{\'e}n, Extreme values in finance, telecommunications,
  and the environment, CRC Press, 2003.

\bibitem{ivette2008tail}
M.~Ivette~Gomes, L.~De~Haan, L.~H. Rodrigues, Tail index estimation for
  heavy-tailed models: accommodation of bias in weighted log-excesses, Journal
  of the Royal Statistical Society: Series B (Statistical Methodology) 70~(1)
  (2008) 31--52.

\bibitem{Buitendag2018}
S.~Buitendag, J.~Beirlant, T.~de~Wet, {Ridge regression estimators for the
  extreme value index}, Extremes 22~(2) (2018) 271--292.
\newblock \href {https://doi.org/10.1007/s10687-018-0338-4}
  {\path{doi:10.1007/s10687-018-0338-4}}.

\bibitem{brazauskas2000robust}
V.~Brazauskas, R.~Serfling, Robust and efficient estimation of the tail index
  of a single-parameter pareto distribution, North American Actuarial Journal
  4~(4) (2000) 12--27.

\bibitem{wang2016new}
C.~Wang, G.~Chen, A new hybrid estimation method for the generalized pareto
  distribution, Communications in Statistics-Theory and Methods 45~(14) (2016)
  4285--4294.

\bibitem{tripathi2014improved}
Y.~M. Tripathi, S.~Kumar, C.~Petropoulos, Improved estimators for parameters of
  a pareto distribution with a restricted scale, Statistical Methodology 18
  (2014) 1--13.

\bibitem{Hill1975}
M.~B. Hill, {A simple general approach to inference about the tail of a
  distribution}, Annals of Statistics 3~(5) (1975) 1163 -- 1174.

\bibitem{Beirlant2004}
J.~Beirlant, Y.~Goegebeur, J.~Teugels, J.~Segers, {Statistics of Extremes :
  Theory and Applications Statistics of Extremes}, 2004.

\bibitem{csorgo1985kernel}
S.~Csorgo, P.~Deheuvels, D.~Mason, Kernel estimates of the tail index of a
  distribution, The Annals of Statistics (1985) 1050--1077.

\bibitem{danielsson1996method}
J.~Danielsson, D.~W. Jansen, C.~G. De~vries, The method of moments ratio
  estimator for the tail shape parameter, Communications in Statistics-Theory
  and Methods 25~(4) (1996) 711--720.

\bibitem{beirlant1996excess}
J.~Beirlant, P.~Vynckier, J.~L. Teugels, Excess functions and estimation of the
  extreme-value index, Bernoulli (1996) 293--318.

\bibitem{beran2014harmonic}
J.~Beran, D.~Schell, M.~Stehl{\'\i}k, The harmonic moment tail index estimator:
  asymptotic distribution and robustness, Annals of the Institute of
  Statistical Mathematics 66~(1) (2014) 193--220.

\bibitem{brilhante2013simple}
M.~F. Brilhante, M.~I. Gomes, D.~Pestana, A simple generalisation of the hill
  estimator, Computational Statistics \& Data Analysis 57~(1) (2013) 518--535.

\bibitem{paulauskas2017class}
V.~Paulauskas, M.~Vai{\v{c}}iulis, A class of new tail index estimators, Annals
  of the Institute of Statistical Mathematics 69~(2) (2017) 461--487.

\bibitem{paulauskas2013improvement}
V.~Paulauskas, M.~Vai{\v{c}}iulis, On an improvement of hill and some other
  estimators, Lithuanian Mathematical Journal 53~(3) (2013) 336--355.

\bibitem{resnick1997smoothing}
S.~Resnick, C.~St{\u{a}}ric{\u{a}}, Smoothing the hill estimator, Advances in
  Applied Probability 29~(1) (1997) 271--293.

\bibitem{Caeiro2005a}
F.~Caeiro, M.~I. Gomes, D.~Pestana, {Direct Reduction of Bias of the Classical
  Hill Estimator}, Statistical 3~(2) (2005) 113--136.

\bibitem{Beirlant1999}
J.~Beirlant, G.~Dierckx, Y.~Goegebeur, G.~Matthys, {Tail Index Estimation and
  an Exponential Regression Model}, Extremes 2~(2) (1999) 177--200.
\newblock \href {https://doi.org/10.1023/A:1009975020370}
  {\path{doi:10.1023/A:1009975020370}}.

\bibitem{Beirlant2002}
J.~Beirlant, G.~Dierckx, A.~Guillou, C.~Starica, {On exponential
  representations of log-spacings of extreme order statistics}, Extremes 5
  (2002) 157--180.

\bibitem{ocran2021reduced}
E.~Ocran, R.~Minkah, K.~Doku-Amponsah, A reduced-bias weighted least square
  estimation of the extreme value index, arXiv preprint arXiv:2110.08570
  (2021).

\bibitem{gomesa2002asymptotically}
M.~I. Gomes, M.~J. Martins, “asymptotically unbiased” estimators of the
  tail index based on external estimation of the second order parameter,
  Extremes 5~(1) (2002) 5--31.

\end{thebibliography}


\begin{thebibliography}{}

\bibitem[Beirlant et~al., 1999]{Beirlant1999}
Beirlant, J., Dierckx, G., Goegebeur, Y., and Matthys, G. (1999).
\newblock {Tail Index Estimation and an Exponential Regression Model}.
\newblock {\em Extremes}, 2(2):177--200.

\bibitem[Beirlant et~al., 2002]{Beirlant2002}
Beirlant, J., Dierckx, G., Guillou, A., and Starica, C. (2002).
\newblock {On exponential representations of log-spacings of extreme order
  statistics}.
\newblock {\em Extremes}, 5:157--180.

\bibitem[Beirlant et~al., 2004]{Beirlant2004}
Beirlant, J., Goegebeur, Y., Teugels, J., and Segers, J. (2004).
\newblock {\em {Statistics of Extremes : Theory and Applications Statistics of
  Extremes}}.

\bibitem[Beirlant et~al., 1996]{beirlant1996excess}
Beirlant, J., Vynckier, P., and Teugels, J.~L. (1996).
\newblock Excess functions and estimation of the extreme-value index.
\newblock {\em Bernoulli}, pages 293--318.

\bibitem[Beran et~al., 2014]{beran2014harmonic}
Beran, J., Schell, D., and Stehl{\'\i}k, M. (2014).
\newblock The harmonic moment tail index estimator: asymptotic distribution and
  robustness.
\newblock {\em Annals of the Institute of Statistical Mathematics},
  66(1):193--220.

\bibitem[Brilhante et~al., 2013]{brilhante2013simple}
Brilhante, M.~F., Gomes, M.~I., and Pestana, D. (2013).
\newblock A simple generalisation of the hill estimator.
\newblock {\em Computational Statistics \& Data Analysis}, 57(1):518--535.

\bibitem[Buitendag et~al., 2018]{Buitendag2018}
Buitendag, S., Beirlant, J., and de~Wet, T. (2018).
\newblock {Ridge regression estimators for the extreme value index}.
\newblock {\em Extremes}, 22(2):271--292.

\bibitem[Caeiro et~al., 2005]{Caeiro2005a}
Caeiro, F., Gomes, M.~I., and Pestana, D. (2005).
\newblock {Direct Reduction of Bias of the Classical Hill Estimator}.
\newblock {\em Statistical}, 3(2):113--136.

\bibitem[Csorgo et~al., 1985]{csorgo1985kernel}
Csorgo, S., Deheuvels, P., and Mason, D. (1985).
\newblock Kernel estimates of the tail index of a distribution.
\newblock {\em The Annals of Statistics}, pages 1050--1077.

\bibitem[Danielsson et~al., 1996]{danielsson1996method}
Danielsson, J., Jansen, D.~W., and De~vries, C.~G. (1996).
\newblock The method of moments ratio estimator for the tail shape parameter.
\newblock {\em Communications in Statistics-Theory and Methods},
  25(4):711--720.

\bibitem[Finkenstadt and Rootz{\'e}n, 2003]{finkenstadt2003extreme}
Finkenstadt, B. and Rootz{\'e}n, H. (2003).
\newblock {\em Extreme values in finance, telecommunications, and the
  environment}.
\newblock CRC Press.

\bibitem[Fraga~Alves, 2001]{alves2001location}
Fraga~Alves, M. (2001).
\newblock A location invariant hill-type estimator.
\newblock {\em Extremes}, 4(3):199--217.

\bibitem[Gkillas and Katsiampa, 2018]{gkillas2018application}
Gkillas, K. and Katsiampa, P. (2018).
\newblock An application of extreme value theory to cryptocurrencies.
\newblock {\em Economics Letters}, 164:109--111.

\bibitem[Gomes and Martins, 2002]{gomesa2002asymptotically}
Gomes, M.~I. and Martins, M.~J. (2002).
\newblock “asymptotically unbiased” estimators of the tail index based on
  external estimation of the second order parameter.
\newblock {\em Extremes}, 5(1):5--31.

\bibitem[Hill, 1975]{Hill1975}
Hill, M.~B. (1975).
\newblock {A simple general approach to inference about the tail of a
  distribution}.
\newblock {\em Annals of Statistics}, 3(5):1163 -- 1174.

\bibitem[Kithinji et~al., 2021]{kithinji2021adjusted}
Kithinji, M.~M., Mwita, P.~N., and Kube, A.~O. (2021).
\newblock Adjusted extreme conditional quantile autoregression with application
  to risk measurement.
\newblock {\em Journal of Probability and Statistics}, 2021.

\bibitem[Longin, 2016]{longin2016extreme}
Longin, F. (2016).
\newblock {\em Extreme events in finance: A handbook of extreme value theory
  and its applications}.
\newblock John Wiley \& Sons.

\bibitem[Mehrnia and Coleri, 2021]{mehrnia2021wireless}
Mehrnia, N. and Coleri, S. (2021).
\newblock Wireless channel modeling based on extreme value theory for
  ultra-reliable communications.
\newblock {\em IEEE Transactions on Wireless Communications}.

\bibitem[Minkah, 2020]{minkah2020tail}
Minkah, R. (2020).
\newblock Tail index estimation of the generalised pareto distribution using a
  pivot from a transformed pareto distribution.
\newblock {\em Science and Development Journal}, 4(1):1--19.

\bibitem[Minkah et~al., 2021]{minkah2021robust}
Minkah, R., de~Wet, T., and Ghosh, A. (2021).
\newblock Robust estimation of pareto-type tail index through an exponential
  regression model.
\newblock {\em Communications in Statistics-Theory and Methods}, pages 1--19.

\bibitem[Paulauskas and Vai{\v{c}}iulis, 2013]{paulauskas2013improvement}
Paulauskas, V. and Vai{\v{c}}iulis, M. (2013).
\newblock On an improvement of hill and some other estimators.
\newblock {\em Lithuanian Mathematical Journal}, 53(3):336--355.

\bibitem[Paulauskas and Vai{\v{c}}iulis, 2017]{paulauskas2017class}
Paulauskas, V. and Vai{\v{c}}iulis, M. (2017).
\newblock A class of new tail index estimators.
\newblock {\em Annals of the Institute of Statistical Mathematics},
  69(2):461--487.

\bibitem[Resnick and St{\u{a}}ric{\u{a}}, 1997]{resnick1997smoothing}
Resnick, S. and St{\u{a}}ric{\u{a}}, C. (1997).
\newblock Smoothing the hill estimator.
\newblock {\em Advances in Applied Probability}, 29(1):271--293.

\bibitem[Rohrbeck et~al., 2018]{rohrbeck2018extreme}
Rohrbeck, C., Eastoe, E.~F., Frigessi, A., and Tawn, J.~A. (2018).
\newblock Extreme value modelling of water-related insurance claims.
\newblock {\em The Annals of Applied Statistics}, 12(1):246--282.

\end{thebibliography}

\end{document}